\documentclass[10pt,fullpage]{article}

\usepackage{fullpage}
\usepackage{graphicx}
\usepackage{balance}  % for  \balance command ON LAST PAGE 
\usepackage{times}
\usepackage{amsmath}
\usepackage{amssymb}
\usepackage{amstext}
\usepackage{amsthm}
\usepackage{epsfig}
\usepackage{placeins}
\usepackage{subfig}
\usepackage{paralist}
\usepackage{caption}
\usepackage{multirow}
\usepackage{algorithm}
\usepackage{algorithmic}
\usepackage{comment}
\usepackage{url}
\usepackage{balance}
\usepackage{enumitem}

\newcommand{\eat}[1]{}

\newcommand{\Var}[1]{\text{Var}\left(#1\right)}

\makeatletter
\newenvironment{sql}%
 {\vskip 5pt\begin{list}{}{%
  \setlength{\topsep}{0pt}\setlength{\partopsep}{0pt}\setlength{\parskip}{0pt}%
  \setlength{\parsep}{0pt}\setlength{\labelwidth}{0pt}%
  \setlength{\rightmargin}{0pt}\setlength{\leftmargin}{0pt}%
  \setlength{\labelsep}{0pt}%
  \obeylines\@vobeyspaces\normalfont\ttfamily%
  \item[]}}
 {\end{list}\vskip5pt\noindent}
\makeatother

\newtheorem{thm}{Theorem}

\begin{document}

\date{February 2017}

\title{OLA-RAW: Scalable Exploration over Raw Data}

\author{
Yu Cheng \hspace*{2cm} Weijie Zhao \hspace*{2cm} Florin Rusu\\
       \small{University of California Merced}\\
       \small{5200 N Lake Road}\\
       \small{Merced, CA 95343}\\
       \small\texttt{\{ycheng4, wzhao23, frusu\}@ucmerced.edu}
}

\maketitle

%%%%%%%%%%%%%%%%%%%%%%%%%%%%%%%%%%%%%%%%%%%%%%%%%%%%%%%%
%\input{abstract}
\begin{abstract}

In-situ processing has been proposed as a novel data exploration solution in many domains generating massive amounts of raw data, e.g., astronomy, since it provides immediate SQL querying over raw files. The performance of in-situ processing across a query workload is, however, limited by the speed of full scan, tokenizing, and parsing of the entire data. Online aggregation (OLA) has been introduced as an efficient method for data exploration that identifies uninteresting patterns faster by continuously estimating the result of a computation during the actual processing---the computation can be stopped as early as the estimate is accurate enough to be deemed uninteresting. However, existing OLA solutions have a high upfront cost of randomly shuffling and/or sampling the data.
In this paper, we present OLA-RAW, a bi-level sampling scheme for parallel online aggregation over raw data. Sampling in OLA-RAW is query-driven and performed exclusively in-situ during the runtime query execution, without data reorganization. This is realized by a novel resource-aware bi-level sampling algorithm that processes data in random chunks concurrently and determines adaptively the number of sampled tuples inside a chunk. In order to avoid the cost of repetitive conversion from raw data, OLA-RAW builds and maintains a memory-resident bi-level sample synopsis incrementally. We implement OLA-RAW inside a modern in-situ data processing system and evaluate its performance across several real and synthetic datasets and file formats. Our results show that OLA-RAW chooses the sampling plan that minimizes the execution time and guarantees the required accuracy for each query in a given workload. The end result is a focused data exploration process that avoids unnecessary work and discards uninteresting data.

\end{abstract}

%%%%%%%%%%%%%%%%%%%%%%%%%%%%%%%%%%%%%%%%%%%%%%%%%%%%%%%%
%\input{introduction}
\section{INTRODUCTION}\label{sec:intro}

In the era of data deluge, massive amounts of \textit{raw data} are generated at an unprecedented scale by mobile applications, sensors, and scientific experiments. The vast majority of these read-only data are stored as application-specific files containing millions -- if not billions -- of records. \textit{Data exploration} is the initial step in extracting knowledge from these data. Aggregate statistics are computed in order to assess the quality of the raw data, before a thorough investigation on transformed data is performed. The main goal of data exploration is to determine if the time-consuming data transformation and in-depth analysis are necessary. Thus, data exploration does not have to be exact. As long as accurate \textit{estimates} that guide the decision process are generated, its goal is achieved. This allows for an extensive set of optimization strategies that reduce I/O and CPU utilization to be employed. However, if the detailed analysis is triggered, the work performed during exploration should allow for \textit{incremental} extension. In order to illustrate these concepts, we provide an example from a real application in astronomy.

\textbf{Motivating example.}
The Palomar Transient Factory\footnote{\url{www.astro.caltech.edu/ptf/}} (PTF) project~\cite{ptf:overview} aims to identify and automatically classify transient astrophysical objects such as variable stars and supernovae in real-time. A list of potential transients -- or candidates -- is extracted from the images taken by the telescope during a night. They are stored as a table in one or more FITS\footnote{\url{http://heasarc.gsfc.nasa.gov/fitsio/}} files. The initial stage in the identification process is to execute a series of aggregate queries over the batch of extracted candidates. This corresponds to data exploration. The general \texttt{SQL} form of the queries is:

%%%%%%%%%%%%%%%%%%%%%%%%%%%%%%%%%%%%%%%%%%%%%%%%
\begin{minipage}{.45\textwidth}
\begin{sql}
SELECT AGGREGATE(expression) AS agg
FROM candidate
WHERE predicate
HAVING agg < threshold
\end{sql}
\end{minipage}\hfill
%%%%%%%%%%%%%%%%%%%%%%%%%%%%%%%%%%%%%%%%%%%%%%%%
\begin{minipage}{.5\textwidth}
where \textit{AGGREGATE} is SUM, COUNT, or AVERAGE and \textit{threshold} is a verification parameter. These queries check certain statistical properties of the entire batch and are executed in sequence---a query is executed only if all the previous queries are satisfied.
\end{minipage}\hfill
If the candidate batch passes the verification criteria, an in-depth analysis is performed for individual candidates. The entire process -- verification and in-depth analysis -- is executed by querying a PostgreSQL\footnote{\url{http://www.postgresql.org/}} database---only after the candidates are loaded from the original FITS files. This workflow is highly inefficient for two reasons. First, the verification cannot start until data are loaded. Second, if the batch does not pass the verification, both the time spent for loading and the storage used for data replication are wasted.

%%%%%%%%%%%
\begin{figure*}[htbp]
\begin{minipage}[htbp]{.32\textwidth}
	\centering
	\includegraphics[width=\textwidth]{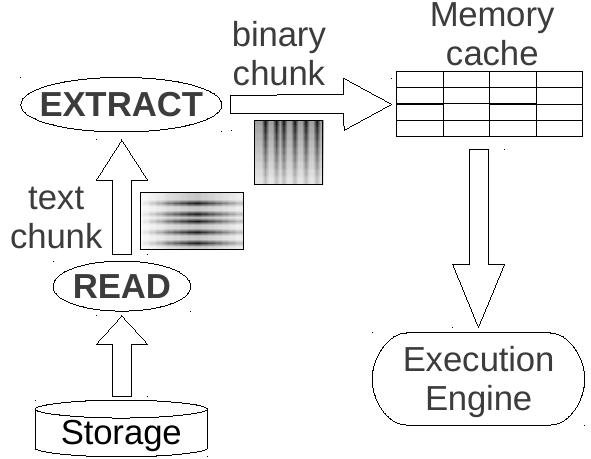}
	\caption{Raw data processing.}
	\label{fig:scanraw}
\end{minipage}\hfill
\begin{minipage}[htbp]{.32\textwidth}
	\centering
	\includegraphics[width=\textwidth]{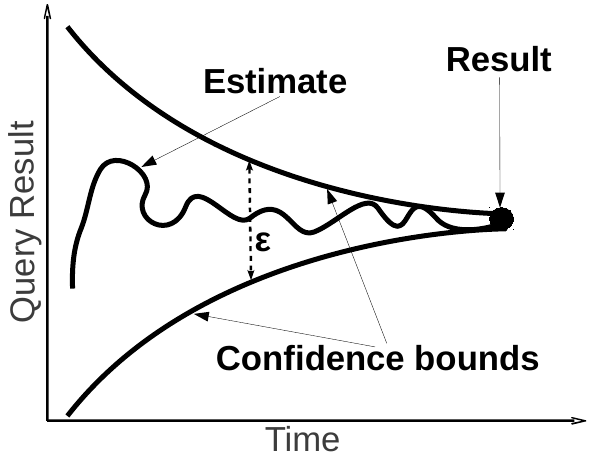}
	\caption{Online aggregation.}
	\label{fig:ola}
\end{minipage}\hfill
\begin{minipage}[htbp]{.32\textwidth}
	\centering
	\includegraphics[width=\textwidth]{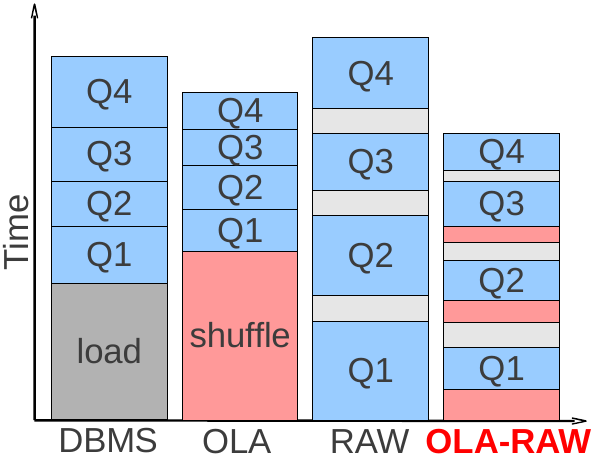}
	\caption{OLA-RAW approach.}
	\label{fig:ola-raw:high-level}
\end{minipage}\hfill
\end{figure*}
%%%%%%%%%%%

%%%%%%%%%%%%%%%%%%%%%%%%%%%%%%%%%%%%%%%%%%%%%%%%%%%%%%%%
\textbf{Raw data processing.}
To reduce the high upfront database loading cost, multiple raw data processing systems have been recently introduced~\cite{files-queries-results,nodb,instant-loading,data-vaults,scanraw,invisible-loading}. They are extensions of the external table mechanism supported by standard database servers~\cite{oracle:external-tables,mysql:external-tables}. These systems execute \texttt{SQL} queries directly over raw data while optimizing the conversion process into the format required by the query engine (Figure~\ref{fig:scanraw}). This eliminates loading and provides instant access to data---verification can start immediately in our example. However, since verification consists of more than a single query, the overall time incurred by raw data processing can be larger than in a database because full access to the raw data is required for every query. Several systems~\cite{nodb,invisible-loading,scanraw} provide a dynamic tradeoff between the time to access data and the query execution time by adaptively loading a portion of the data during processing. This allows for gradually improved query execution times while reducing the amount of storage for replication. However, these systems are data-agnostic and cannot identify uninteresting patterns early in the processing. This results in wasted CPU and storage resources.

\begin{comment}
The existing solutions have severe limitations when applied to data exploration. Standard databases require data loading even for a single query over a raw file---the case in many data exploration tasks. Although external tables avoid loading, the system has to access the entire data for every query, which causes poor performance over a sequence of interesting queries. NoDB, invisible loading~\cite{invisible-loading}, and SCANRAW are query-driven in-situ processing systems that improve their performance gradually and converge to the database execution time by loading all the data. However, they are \textit{data-agnostic} and cannot identify uninteresting patterns early in the processing. This results in wasted CPU and storage resources.
\end{comment}

%%%%%%%%%%%%%%%%%%%%%%%%%%%%%%%%%%%%%%%%%%%%%%%%%%%%%%%%
\textbf{Online aggregation.}
Since the goal of data exploration -- batch verification in our example -- is only to determine if the individual candidate in-depth analysis is necessary, it is not mandatory to evaluate each query in the sequence to completion. As early as the relationship between the aggregate and the verification threshold can be accurately inferred, the query can be stopped. This relationship can be determined by using only an estimate of the aggregate. If the aggregate \texttt{agg} -- or its estimate -- is larger than the threshold, the verification fails and no in-depth analysis is required. Otherwise, we can proceed to the subsequent query in the verification.
Online aggregation (OLA)~\cite{ola,dbo,AQP-book} provides a sound framework to reason about the aggregate estimation involved in verification. The main idea in OLA is to estimate the query result based on a sample of the data. In addition to the estimator, OLA defines a principled method to derive confidence bounds that permit the correct identification of the relationship with the threshold. The estimator and the bounds are computed from a sample much smaller in size than the overall dataset, thus reducing the execution time of a verification query tremendously (Figure~\ref{fig:ola}).
The existing OLA solutions have strict requirements imposed by the sampling procedure. Data shuffling~\cite{demo:dbo,continuous-sampling,pfola:dapd} is considered the standard procedure to extract samples of increasing size from a dataset. Shuffling generates a permutation of the data as a query preprocessing step such that a runtime sequential scan results in random samples of increasing size. However, shuffling creates a secondary copy of the data and incurs significant processing time---even more than loading.

%%%%%%%%%%%%%%%%%%%%%%%%%%%%%%%%%%%%%%%%%%%%%%%%%%%%%%%%
\textbf{Problem \& approach.}
At high-level, our objective is to \textit{optimally execute exploration over raw data in a shared-memory multi-core environment where I/O operations are overlapped with extraction and several chunks can be processed concurrently while minimizing resource utilization}. In our concrete example, this corresponds to low execution time for the verification process without incurring any loading cost.

%%%%%%%%%%%%%%%%%%%%%%%%%%%%%%%%%%%%%%%%%%%%%%%%%%%%%%%%
%\textbf{Approach.}
Our approach is to seamlessly integrate online aggregation into raw data processing such that we cumulate their benefits. Figure~\ref{fig:ola-raw:high-level} illustrates the intuition behind the proposed OLA-RAW solution with respect to standard database processing (DBMS), online aggregation (OLA), and raw data processing with adaptive loading (RAW), respectively. Similar to RAW, OLA-RAW distributes data loading across the query workload. Notice, though, that loading in RAW -- and, by extension, in OLA-RAW -- corresponds to caching data in memory, not necessarily materializing on secondary storage. The same idea is extended to shuffling. Instead of randomly permuting all the data before performing online aggregation, OLA-RAW partitions shuffling across the queries in the workload. Moreover, loading and shuffling are combined incrementally such that loaded data do not require further shuffling. Essentially, \textit{OLA-RAW provides a resource-aware parallel mechanism to adaptively extract and incrementally maintain samples from raw data}. The end goal is to reduce the high upfront cost of loading (DBMS) and shuffling (OLA), and to minimize the amount of data accessed by RAW, as long as estimates are accepted by the user---the general situation in data exploration.

%%%%%%%%%%%%%%%%%%%%%%%%%%%%%%%%%%%%%%%%%%%%%%%%%%%%%%%%
\textbf{Challenges.}
The realization of OLA-RAW poses a series of difficult challenges. First and foremost, an efficient sampling mechanism targeted at raw data has to be devised. The sampling mechanism has to work in-place, over data in the original format. It has to minimize the amount of raw data read and/or extracted into the processing representation since these are the fundamental limitations of raw data processing. Given our focus on parallel processing, the sampling mechanism has to cope with the so called ``inspection paradox''~\cite{online-mapreduce}. Since the estimate is correlated with the extraction time, the order in which chunks are considered has to be the same with the extraction scheduling order. The second challenge is defining and analyzing estimators for the sampling mechanism. In order to be amenable to online aggregation, the estimators have to be integrated in the sampling mechanism and they have to support incremental computation over samples of increasing size. A third challenge corresponds to the incremental maintenance of samples. Since extracting samples from raw data is expensive, a mechanism that preserves them in memory for further use in subsequent queries and maintains them incrementally is necessary. This has to be realized efficiently---the goal is to compute the estimate as fast as possible, not to maintain the sample. From an implementation perspective, the integration of online aggregation into a resource-aware raw data processing system is always challenging because of the complex interactions between I/O, extraction, sampling, and estimation.

\begin{comment}
We consider the problem of \textit{efficiently producing estimation result during queries execution in-situ over raw files}. Simply extend the classic OLA work on top of raw file process system is not suffice. Since in order to maximize the I/O utilization, raw file processing system usually sequentially scan the file to produce data, which cannot be directly used by classic OLA model. Our objective is to design an novel system --combining advantages from traditional OLA and in-situ data processing system-- that provides instant access to data and also produce estimation result during query execution, besides the system should achieve optimal performance when the workload consists of a sequence of queries. There are two aspects to this problem. First, methods that provide sample-based query execution over raw files have to be developed. And second, a mechanism for query-driven gradual sample maintenance has to be devised. This mechanism interferes minimally -- if at all -- with normal query processing and guarantees to speed up the following queries.
\end{comment}

%%%%%%%%%%%%%%%%%%%%%%%%%%%%%%%%%%%%%%%%%%%%%%%%%%%%%%%%
\textbf{Contributions.}
The main contribution of this paper is a novel bi-level sampling scheme for OLA-RAW that addresses the aforementioned challenges. OLA-RAW sampling is query-driven and performed exclusively in-situ during query execution, without data reorganization. In order to avoid the expensive conversion cost, OLA-RAW builds and maintains incrementally a memory-resident bi-level sample synopsis. These are achieved through a series of technical contributions detailed in the following:
\begin{compactitem}
\item We define online aggregation over raw data in a multi-core shared-memory setting (Section~\ref{sec:preliminaries}).
\item We provide a parallel chunk-level sampling mechanism that avoids the inherent inspection paradox (Section~\ref{sec:chunk-sampling}).
\item We introduce a novel parallel bi-level sampling scheme that supports continuous estimation -- not only at chunk boundaries -- during the online aggregation process (Section~\ref{sec:bilevel-sampling}).
\item We devise a resource-aware policy to determine the optimal chunk sample size for the proposed bi-level sampling scheme (Section~\ref{sec:ola-raw-sampling}).
\item We design a memory-resident bi-level sample synopsis that is built and maintained incrementally following a variance-driven strategy (Section~\ref{sec:synopsis}).
\end{compactitem}
We implement OLA-RAW inside a state-of-the-art in-situ data processing system and evaluate its performance across several real and synthetic datasets and file formats. We investigate the importance of each OLA-RAW component as well as the overall solution. Our results (Section~\ref{sec:experiments}) show that OLA-RAW chooses the sampling plan that minimizes the execution time and guarantees the required accuracy for each query in a given workload.

\section{PRELIMINARIES}\label{sec:preliminaries}

In this section, we introduce query processing over raw data and online aggregation. We define the online aggregation over raw data problem and identify the challenges to be addressed by a coherent solution that integrates the two.

%%%%%%%%%%%%%%%%%%%%%%%%%%%%%%%%%%%%%%%%%%%%%%%%%%%%%%%%%
\subsection{Parallel Raw Data Processing}\label{sec:prelim:in-situ}

Raw data processing is depicted in Figure~\ref{fig:scanraw}. The input to the process is a raw file from a non-volatile storage device, e.g., disk or SSD, a schema that can include optional attributes, and a procedure to extract tuples with the given schema from the raw file. The output is a tuple representation that can be processed by the query engine and, possibly, is cached in memory. In the \texttt{READ} stage, data are read from the original raw file, chunk-by-chunk, using the file system's functionality. A chunk contains multiple records and represents the unit of processing. Without additional information about the structure or the content -- stored inside the file or in some external structure -- the entire file has to be read the first time it is accessed. \texttt{EXTRACT} transforms tuples from raw format into the processing representation based on the schema provided and using the extraction procedure given as input to the process. There are two main tasks in \texttt{EXTRACT}. The first is to identify the schema attributes and output a vector containing the starting position for every attribute in the tuple---or a subset, if the query does not access all the attributes. Second, attributes are converted from the raw format to their corresponding binary type and mapped to the processing representation of the tuple---the record in a row-store, or the array in column-stores, respectively. At the end of \texttt{EXTRACT}, data are loaded in memory and ready for query processing. In this paper, we consider parallel raw data processing in the context of the SCANRAW operator~\cite{scanraw,scanraw:tods} and NoDB~\cite{nodb}. SCANRAW overlaps the I/O operations with \texttt{EXTRACT} over multiple chunks in a super-scalar pipeline architecture, i.e., multiple chunks are extracted concurrently. NoDB caches binary chunks in memory to avoid subsequent extraction.

\begin{comment}
Multiple paths can be taken at this point. In standard database loading, data are first written to the database and only then query processing starts. In external tables~\cite{mysql:external-tables,oracle:external-tables}, data are passed to the query engine and discarded afterwards. In NoDB~\cite{files-queries-results,nodb} and in-memory databases~\cite{instant-loading,data-vaults}, data are kept in memory for subsequent processing. SCANRAW~\cite{scanraw} overlaps the I/O operations with \texttt{EXTRACT} over multiple chunks in a super-scalar pipeline architecture. Moreover, the interaction between \texttt{READ} and \texttt{WRITE} is carefully scheduled in order to avoid interference.
\end{comment}

%%%%%%%%%%%%%%%%%%%%%%%%%%%%%%%%%%%%%%%%%%%%%%%%%%%%%%%%%
\subsection{Online Aggregation on Raw Data}\label{sec:prelim:online-agg}

\begin{comment}
The main idea in online aggregation (OLA) is to compute only an estimate of the query result (Figure~\ref{fig:ola}) based on a sample of the data~\cite{ola}. In order to provide any useful information, though, the estimate is required to be accurate and statistically significant. Different from one-time estimation~\cite{AQP-book,TamingTerabytes} that might produce very inaccurate estimates for arbitrary queries, OLA is an iterative process in which a series of estimators with improving accuracy are generated. This is accomplished by including more data in estimation, i.e., increasing the sample size, from one iteration to another. As more data are processed towards computing the final result, the accuracy of the estimator improves accordingly. For this to be true, though, data are required to be processed in a statistically meaningful order, i.e., random order, to allow for the definition and analysis of the estimator. The user can decide to stop the query or to run a subsequent iteration based on the accuracy $\epsilon$, i.e., confidence bounds width, of the estimator. As long as the bounds shrink fast enough, the time to execute the entire process is expected to be shorter than computing the exact result over the entire dataset.
\end{comment}

We consider online aggregation over a single table $T$ stored in some arbitrary sequential raw format, e.g., CSV, JSON, or FITS, and general aggregate queries of the form:

%%%%%%%%%%%%%%%%%%%%%%%%%%%%%%%%%%%%%%%%%%%%%%%%
\begin{minipage}{.4\textwidth}
\begin{sql}
SELECT AGGREGATE(expression)
FROM T
WHERE predicate
\end{sql}
\end{minipage}\hfill
%%%%%%%%%%%%%%%%%%%%%%%%%%%%%%%%%%%%%%%%%%%%%%%%
\begin{minipage}{.55\textwidth}
where \textit{AGGREGATE} is one of SUM, COUNT, or AVERAGE and \textit{expression} is a numeric expression, such as $T.a$ or $(T.a-T.b)^{2}$, that involves one or more columns of $T$. \end{minipage}\hfill
These aggregation functions are the most commonly used in practice. Online aggregation over a single raw data source is a fundamental problem arising not only in queries that explicitly involve a single table -- this is the standard scenario in raw data processing -- but also in queries on ``star'' schemas that consist of a massive ``fact'' table -- which is sampled -- and many smaller ``dimension'' tables---which are cached in memory. In general, the proposed methods apply to queries that involve joins between multiple tables, provided that sampling is performed on exactly one raw data source and each join attribute is a foreign key. \textit{GROUP BY} queries can also be handled using the methods in this paper by simply treating each group as a separate query and running all the queries simultaneously. A group-specific version of the \textit{predicate} that accepts only tuples from that particular group is required for each separate query.

In addition to the query, an online aggregation user is typically required to specify two parameters. The \textit{accuracy $\epsilon$} determines when the query can be stopped. Different from one-time estimation~\cite{AQP-book,TamingTerabytes} that might produce inaccurate estimates not satisfying a given $\epsilon$, OLA is an iterative process in which a series of estimators with improving accuracy are generated. This is accomplished by including more data in estimation, i.e., increasing the sample size, from one iteration to another. As more data are processed towards computing the result, the accuracy of the estimator improves accordingly (Figure~\ref{fig:ola}). The \textit{estimation time interval $\delta$} specifies how often the estimate and corresponding confidence bounds are computed. The user can decide to stop the query at any time -- even when the accuracy is not satisfied -- based on the returned estimate and confidence bounds.

Extracting random samples of increasing size at runtime is a complex time-consuming process~\cite{olken-phd}. This is the reason why existing procedures require a certain level of preprocessing. In offline sampling, a series of samples with progressively increasing size -- the largest one being the entire dataset -- are taken, e.g., BlinkDB~\cite{blink}. These samples are stored and processed as independent entities. Offline tuple-based shuffling~\cite{demo:dbo,continuous-sampling,pfola:dapd} guarantees that a runtime sequential scan produces samples of increasing size. In addition to preprocessing time, the offline methods incur a heavy storage overhead. These are in contrast with the requirements of in-situ data processing. Chunk-level~\cite{AQP-book}, i.e., block-level~\cite{online-mapreduce} or cluster~\cite{sampling}, sampling samples over the chunk space instead of the tuple space. This can be done efficiently online -- randomly permute the processing order of chunks -- without any preprocessing. However, chunk-level sampling incurs a higher processing cost because all the tuples inside a chunk have to be included in estimation. While this may be irrelevant for database processing, it is of great importance for in-situ processing due to the \texttt{EXTRACT} stage.

In this paper, we consider chunk-level sampling in the context of parallel raw data processing, specifically the SCANRAW operator. This creates problems because the random chunk order interacts with parallel processing. This, in turn, can trigger the inspection paradox which makes sampling-based estimation impossible. The only solution that addresses this problem in a distributed MapReduce setting is given in~\cite{online-mapreduce}. It defines a multivariate distribution that incorporates several timing parameters in addition to the aggregate chunk value. Since we focus on multi-thread parallelism in a shared-memory setting, the timing parameters are too similar to have a discriminative effect. Moreover, we can take advantage of the centralized shared-memory environment to eliminate the inspection paradox without resorting to expensive distributed synchronization.

\section{CHUNK-LEVEL SAMPLING}\label{sec:chunk-sampling}

In this section, we give a simple solution that extends chunk-level sampling to parallel online aggregation over raw data. As far as we know, this is the first attempt to parallelize chunk-level sampling in a shared-memory setting.

Chunk-level sampling can be implemented similarly to a random data scan. The random order in which chunks are read from storage is determined before query execution starts (Figure~\ref{fig:sampling-generation}). At runtime, chunks are fully processed by computing the aggregate function and any additional statistics required for estimation over all the tuples in the chunk. Essentially, a single value is generated for the entire chunk. For example, if \texttt{SUM(B)} has to be estimated in Figure~\ref{fig:sampling-generation}, the value corresponding to chunk 2 is $110.931$. This is what is used in estimation. At each step, the chunks processed represent a simple random sample without replacement over the space of all chunks in $T$. As more chunks are processed, the size of the sample increases. This guarantees that the accuracy of the estimate improves, i.e., the width of the confidence bounds shrinks. 

Chunk-level sampling does not require any preprocessing, thus it fits perfectly in the in-situ raw data processing setting. It also minimizes the number of chunks read from storage when compared to tuple-based sampling. However, the accuracy of chunk-level sampling can be orders of magnitude lower than that of tuple-based sampling for the same number of processed tuples~\cite{bi-level-bernoulli}. This is of critical importance for in-situ processing because the gap between I/O bandwidth and CPU utilization is less wide than in standard databases---\texttt{EXTRACT} can be a time-consuming CPU task. Consider, for example, processing a JSON file with a deep schema. Extracting all the objects from a chunk can easily become the bottleneck, making the entire process CPU-bound. Since chunk-level sampling requires more tuples to be processed in order to achieve the same accuracy, it is very likely that it also incurs a higher execution time than tuple-based sampling.

Parallel in-situ processing systems overlap \texttt{READ} and \texttt{EXTRACT}, while also performing multiple \texttt{EXTRACT} instances concurrently. As long as sufficient CPU threads are available, multiple chunks can be extracted in parallel. Due to variations in extraction time across chunks, e.g., different number of tuples, length of attributes, predicate selectivity, the order in which chunks are returned by \texttt{EXTRACT} can be different from the predefined random order. This renders chunk-level sampling estimation impossible due to the inspection paradox. The only solution to avoid this problem is to reorder the chunks before they are included in the estimation---add a synchronization barrier after \texttt{EXTRACT}. For example, in Figure~\ref{fig:sampling-generation} chunk-level estimation is possible only after chunk 2 is extracted, even though chunks 3 and 1 have already been processed---they are scheduled after chunk 2. While reordering eliminates the inspection paradox, it can exacerbate the standard error swings characteristic to chunk-level sampling~\cite{bi-level-bernoulli}. Moreover, it can introduce variable-length gaps, i.e., step behavior, in estimation---there can be long time intervals over which there is no change in the estimator. These diminish the applicability of chunk-level sampling to parallel online aggregation over raw data.

%%%%%%%%%%%%%%%%%%%%%%%%%%%%%%%%%%%%
\begin{figure*}[htbp]
\begin{center}
\includegraphics[width=\textwidth]{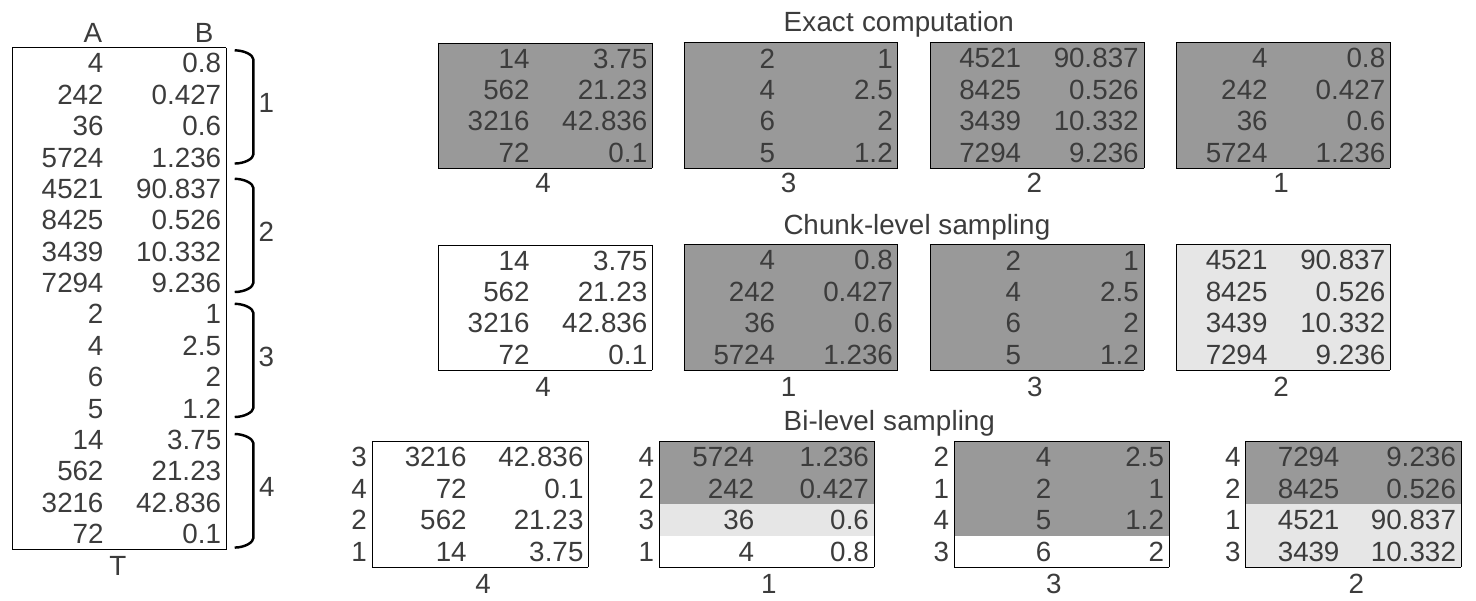}
\caption{Sampling strategies for online aggregation over raw data. The dark gray portions correspond to data included in estimation. The light gray portions correspond to data under processing, but not yet included in estimation. The white portions correspond to unprocessed data. Chunks are scheduled for processing from right to left.}
\label{fig:sampling-generation}
\end{center}
\end{figure*}
%%%%%%%%%%%%%%%%%%%%%%%%%%%%%%%%%%%%%

%%%%%%%%%%%%%%%%%%%%%%%%%%%%%%%%%%%%%%%%%%%%%%%%%%%%%%%%
%\input{bilevel-solution}
\section{BI-LEVEL SAMPLING}\label{sec:bilevel-sampling}

In this section, we introduce a novel bi-level sampling scheme for parallel online aggregation. The proposed scheme differs significantly from bi-level Bernoulli sampling~\cite{bi-level-bernoulli} -- the only work that applies bi-level sampling to databases we are aware of -- in which the goal is to extract a one-time Bernoulli sample with rate $q$. Our objective is to design an incremental sampling procedure that supports adaptive sample size increase in order to achieve continuous accuracy improvement.

%%%%%%%%%%%%%%%%%%%%%%%%%%%%%%%%%%%%%%%%%%%%%%%%%%%%%%%%%
\subsection{Sequential Procedure}\label{sec:bilevel-sampling:seq-procedure}

We propose the following bi-level sampling scheme. Chunks are read in a predetermined random order similar to chunk-level sampling. However, instead of aggregating all the tuples in the chunk into a single value that becomes a surrogate for the chunk, a secondary sampling process is performed over the tuples in the chunk. This is realized by randomly shuffling the order in which tuples are extracted. Independent orders are used across chunks. Since this is done in memory, no significant overhead is incurred. Figure~\ref{fig:sampling-generation} depicts the entire procedure.

At each step during the process, the set of sampled tuples correspond to a bi-level -- or two-stage -- sample without replacement~\cite{sampling-techniques,sampling}. This provides more flexibility over chunk-level sampling since estimates can be computed at any point in the process---not only at chunk boundaries. Moreover, the sample size can be increased either by including more chunks or more tuples inside a chunk. It is important to notice that bi-level sampling degenerates to chunk-level sampling when all the tuples inside a chunk are included in the sample. This is likely to happen when there is high variability inside a chunk. When tuples inside a chunk are similar, however, bi-level sampling allows for the chunk to be represented by a much smaller number of tuples. This has the potential to dramatically reduce the cost of \texttt{EXTRACT} in raw data processing. The downside of bi-level sampling is the larger number of parameters. In addition to the number of chunks, the number of sampled tuples inside each chunk has to be specified. However, these are determined dynamically in online aggregation, thus, this is not a real problem.

%%%%%%%%%%%%%%%%%%%%%%%%%%%%%%%%%%%%%%%%%%%%%%%%%%%%%%%%%
\subsection{Parallel Procedure}\label{sec:bilevel-sampling:par-procedure}

As in the case of chunk-level sampling, when bi-level sampling is applied to parallel online aggregation over raw data, the inspection paradox can invalidate the entire sampling procedure. The impact of the inspection paradox is further aggravated by the different number of tuples extracted across chunks. For example, consider a highly-variable chunk followed by a uniform one (chunk 2 and 3 in Figure~\ref{fig:sampling-generation}). Although chunk 2 is scheduled before chunk 3, fewer tuples from chunk 3 have to be extracted and included in estimation. As a result, chunk 3 -- even chunk 1 in Figure~\ref{fig:sampling-generation} -- is extracted before chunk 2. In the case of chunk-level sampling, an estimate can be generated only after chunk 2 finishes---the estimate includes chunk 3 and 1, as well. While exactly the same conditions apply to bi-level sampling, the secondary sampling process inside chunks provides additional opportunities not to delay estimation---an essential requirement in online aggregation.

A major contribution of this paper is a novel solution for continuous estimation. We devise a mechanism that enforces the existence of samples from all the chunks in \texttt{EXTRACT} at any estimation time interval $\delta$. Each \texttt{EXTRACT} thread is configured with a timing parameter $t^{\textit{eval}}$ that specifies when samples from the chunk have to be produced. This can happen multiple times during the execution of \texttt{EXTRACT}. Since chunks are scheduled for extraction sequentially, this guarantees that samples are extracted in order. The number of tuples included in the sample, however, can vary based on the properties of the chunk. This is illustrated in Figure~\ref{fig:sampling-generation} where 3 tuples are sampled from chunk 3, while only 2 tuples from chunk 2 and 1, respectively. As long as the timing parameter $t^{\textit{eval}}$ is smaller than $\delta$, improved estimates can be generated. $t^{\textit{eval}}$ is -- in a sense -- related to the variables $t^{\textit{sch}}$ and $t^{\textit{proc}}$ in~\cite{online-mapreduce}. However, instead of including it in the estimation snapshot, we use timing to enforce the bi-level sampling process and its corresponding estimation. The overhead incurred by the timing mechanism is minimal and reduces to inspecting a timer after groups of several tuples -- each tuple, in the extreme -- are extracted.

%%%%%%%%%%%%%%%%%%%%%%%%%%%%%%%%%%%%%%%%%%%%%%%%%%%%%%%%%
\subsection{Estimation}\label{sec:bilevel-sampling:estimation}

We focus on the estimator for the \texttt{SUM} aggregate. \texttt{COUNT} is identical to \texttt{SUM} with $\textit{expression}=1$. As shown in~\cite{bi-level-bernoulli}, only minor modifications have to be made for complex aggregates such as \texttt{AVERAGE}, \texttt{VARIANCE}, or standard deviation. The notation used in our analysis is shown in Table~\ref{tbl:notation}. It adapts the notation for Bernoulli sampling used in~\cite{bi-level-bernoulli} to sampling without replacement.

%%%%%%%%%%%%%%%%%%%%%%%%%%%%%%%%%%
\begin{table}[htbp]
  \begin{center}
    \begin{tabular}{ll}

	\textbf{Symbol} & \textbf{Meaning}\\

	\hline
	\hline
	
	$T$ & Set of tuples in table \\
	$U$ & Set of chunks in table \\
	$C_{j}$ & Set of tuples on chunk $j$ \\

	\hline

	$T'$ & Set of tuples in sample \\
	$U'$ & Set of chunks in sample \\
	$C'_{j}$ & Set of tuples on chunk $j$ that are in sample \\

	\hline

	$M = |T|$ & Number of tuples in table \\
	$N = |U|$ & Number of chunks in table \\
	$M_{j} = |C_{j}|$ & Number of tuples on chunk $j$ \\

	\hline

	$m = |T'|$ & Number of tuples in sample \\
	$n = |U'|$ & Number of chunks in sample \\
	$m_{j} = |C'_{j}|$ & Number of tuples in sample of chunk $j$ \\

	\hline

	$x_{i}$ & Value of \textit{expression} for the $i^{\textit{th}}$ tuple in table ($x_{i} = 0$ if tuple $i$ fails to satisfy \textit{predicate}) \\
	$y_{j} = \sum_{i\in C_{j}} {x_{i}}$ & Sum of $x_{i}$ values on chunk $j$ \\
	$y'_{j} = \sum_{i\in C'_{j}} {x_{i}}$ & Sum of $x_{i}$ values in sample of chunk $j$ \\
	$y''_{j} = \sum_{i\in C'_{j}} {x_{i}^{2}}$ & Sum of $x_{i}^{2}$ values in sample of chunk $j$ \\

	\hline

    \end{tabular}
  \end{center}

\caption{Notation for bi-level sampling used throughout the paper.}\label{tbl:notation}
\end{table}
%%%%%%%%%%%%%%%%%%%%%%%%%%%%%%%%%%

Let $\tau=\sum_{j\in U} {y_{j}} = \sum_{i\in T} {x_{i}}$ denote the true result of the query. In order to define an unbiased estimator for $\tau$, we first introduce an estimator for $y_{j}$, the sum of the expression values in chunk $j$. This is the standard estimator for tuple-level sampling without replacement restricted to a chunk, i.e., $\widehat{y}_{j} = \frac{M_{j}} {m_{j}} y'_{j}$. It is well-known that $\widehat{y}_{j}$ is an unbiased estimator. The unbiased estimator for $\tau$ combines the chunk estimators in a standard sampling without replacement estimator over the chunks:
\begin{equation}\label{eq:estimator}
\widehat{\tau} = \frac{N} {n} \sum_{j=1}^{n} {\widehat{y}_{j}} = \frac{N} {n} \sum_{j=1}^{n} {\frac{M_{j}} {m_{j}} \sum_{i\in C'_{j}} {x_{i}}}
\end{equation}
Confidence bounds are the backbone of online aggregation since they characterize the accuracy of the estimator. This requires the computation of the variance and the definition of a variance estimator based on the samples. By assuming normality based on the Central Limit Theorem~\cite{sampling-techniques}, confidence bounds are computed from the cumulative distribution function (cdf) of the normal distribution. We provide the formulae for the bi-level sampling variance and a corresponding unbiased estimator in the following theorems which are based on~\cite{sampling}.

\begin{thm}\label{thm:variance}
The variance of the bi-level sampling estimator $\widehat{\tau}$ defined in Eq.~\eqref{eq:estimator} is given by:
\begin{equation}\label{eq:variance}
\Var{\widehat{\tau}} = \frac{N} {N-1} \frac{N-n} {n} \sum_{j=1}^{N} {\left( y_{j} - \frac{\sum_{i\in T} {x_{i}}} {N} \right)^{2}} + \frac{N} {n} \sum_{j=1}^{N} \left[ {\frac{M_{j}} {M_{j}-1} \frac{M_{j}-m_{j}} {m_{j}} \sum_{i\in C_{j}} {\left( x_{i} - \frac{y_{j}} {M_{j}} \right)^{2} } } \right]
\end{equation}
\end{thm}
\begin{proof}
This is obtained by a rewriting of the formula Eq.~(11.22, pp. 303) given in~\cite{sampling-techniques}.
\end{proof}

There are two distinct terms in the variance formula---the first for the variance between chunks and the second for the variance inside each chunk. The variance between chunks measures how a chunk deviates from the average across all the chunks. The variance inside a chunk computes the deviation of a tuple from the average value of the chunk. The sampling without replacement nature of the process is reflected in the scaling factors of the two terms. The more chunks are sampled, the smaller the scaling factor, i.e., $\frac{N-n} {n}$. The same applies to the number of tuples sampled inside a chunk. In the extreme case when all the chunks are sampled, the variance across chunks vanishes. This corresponds to stratified sampling. If all the tuples inside a sampled chunk are included, the variance inside that chunk reduces to zero. By ignoring the terms corresponding to non-sampled chunks, the resulting variance corresponds to the chunk-level sampling variance---as well as the estimator $\tau$.

\begin{thm}\label{thm:variance-est}
An unbiased estimator for the variance of bi-level sampling is given by:
\begin{equation}\label{eq:variance-est}
\widehat{\Var{\widehat{\tau}}} = \frac{N} {n} \frac{N-n} {n-1} \sum_{j=1}^{n} {\left( \frac{M_{j}} {m_{j}} y'_{j} - \frac{\sum_{j'=1}^{n} {\frac{M_{j}} {m_{j}} y'_{j'}}} {n} \right)^{2}} + \frac{N} {n} \sum_{j=1}^{n} \left[ {\frac{M_{j}} {m_{j}} \frac{M_{j}-m_{j}} {m_{j}-1} \sum_{i\in C'_{j}} {\left( x_{i} - \frac{y'_{j}} {m_{j}} \right)^{2} } } \right]
\end{equation}
\end{thm}
\begin{proof}
This is obtained by a rewriting of the formula Eq.~(11.24, pp. 303) given in~\cite{sampling-techniques}.
\end{proof}

In order to evaluate the query estimator and estimate its variance, a series of chunk-level statistics have to be computed. They include $m_{j}$, $y'_{j}$, and $y''_{j}$. These quantities require minimum space and processing overhead beyond what is required by the actual query. Additionally, the total number of tuples and the number of tuples in each chunk have to be known. This is problematic for raw data in general. However, since most textual formats store a tuple per line, e.g., CSV and JSON, these values can be easily obtained by a simple execution of the command \texttt{wc -l}. Other formats, e.g., HDF5 and FITS, store these values in the file metadata.

%%%%%%%%%%%%%%%%%%%%%%%%%%%%%%%%%%%%%%%%%%%%%%%%%%%%%%%%
%\input{resource-sampling}
\section{OLA-RAW BI-LEVEL SAMPLING}\label{sec:ola-raw-sampling}

Given a bi-level sample extracted from raw data, Theorem~\ref{thm:variance-est} allows us to derive confidence bounds for an aggregate query. However, the theorem does not specify what is the optimal sampling procedure. In this section, we investigate how to extract a bi-level sample that is optimal for online aggregation over raw data. We define this problem formally and first introduce a solution that guarantees convergence to the required accuracy in one pass over the data. Then, we design a novel resource-aware parallel sampling procedure that allocates the system resources optimally in order to achieve the desired accuracy.

%%%%%%%%%%%%%%%%%%%%%%%%%%%%%%%%%%%%%%%%%%%%%%%%%%%%%%%%%
\subsection{Holistic Sampling}\label{sec:ola-raw-sampling:holistic}

We start with a straightforward realization of the bi-level sampling procedure that does not consider how much to sample from a chunk---it samples the entire chunk. However, in order to avoid the inspection paradox and to reduce the interval between estimations, samples are extracted from each chunk at $t^{\textit{eval}}$ time intervals. We emphasize that this is not chunk-level sampling because any subset of the chunk represents a sample---including the complete chunk. By grouping the samples from all the chunks together, a bi-level sample is generated and an estimate and corresponding confidence bounds can be computed. If the required accuracy $\epsilon$ is satisfied, the query can be stopped. Otherwise, the exact answer is obtained if the query is executed to completion.

%%%%%%%%%%%%%%%%%%%%%%%%%%%%%%%%%%%%%%%%%%%%%%%%%%%%%%%%%
\subsection{Optimization Formulation}\label{sec:ola-raw-sampling:formulation}

The optimization problem for online aggregation over raw data is to minimize the query processing cost under the constraint of achieving the specified accuracy $\epsilon$. In order to define formally this problem, we have to introduce a cost model for bi-level sampling. We start from the cost model for parallel in-situ processing over raw data introduced in~\cite{vert-part:ssdbm-2015}. In this model, the query processing cost is defined as the maximum between the read I/O time $T_{\textit{I/O}}$ and the CPU extraction time with $P$ worker threads $T_{\textit{CPU}}$---\texttt{READ} and \texttt{EXTRACT} are overlapped. In the case of bi-level sampling, $T_{\textit{I/O}}$ is a linear function of the number of sampled chunks $n$, i.e., $T_{\textit{I/O}} \approx n$, while $T_{\textit{CPU}}$ is a linear function of the number of sampled tuples inside the chunk, i.e., $T_{\textit{CPU}} \approx \frac{1}{P} \cdot \sum_{j=1}^{n} {m_{j}}$. The constraint on the accuracy $\epsilon$ can be written as an inequality between the variance estimator $\widehat{\Var{\widehat{\tau}}}$ and a maximum variance $\textit{Var}_{\textit{max}}$ derived from $\epsilon$.

Based on these, the bi-level sampling for online aggregation over raw data optimization problem can be expressed as follows:
\begin{equation}\label{eq:optim-formulation}
\begin{split}
& \text{minimize} \hspace*{0.2cm} \max\left\{ T_{\textit{I/O}}(n), T_{\textit{CPU}}(\overline{m_{j}}) \right\} \text{subject\ to}\\
& \text{constraint} \hspace*{0.2cm} \widehat{\Var{\widehat{\tau}}}(n,\overline{m_{j}}) \leq \textit{Var}_{\textit{max}}
\end{split}
\end{equation}
where the variables to be optimized are $n$ and $\overline{m_{j}}$, $j\in \{1,\dots,n\}$. Closed-form formulae for sequential bi-level Bernoulli sampling are derived in~\cite{bi-level-bernoulli}. They are based on the notion of chunk heterogeneity index which measures the variability of values within a chunk relative to the variability of values between chunks. These formulae cannot be extended to our $\max$ objective function. Moreover, computing the chunk heterogeneity index over raw data in exploration tasks defeats the purpose of in-situ processing. The algorithms proposed in~\cite{bi-level-bernoulli} require the existence of a pilot sample or of complex chunk-level statistics such as the number of distinct values and the variance. These are not available for raw data.

%%%%%%%%%%%%%%%%%%%%%%%%%%%%%%%%%%%%%%%%%%%%%%%%%%%%%%%%%
\subsection{Single-Pass Sampling}\label{sec:ola-raw-sampling:single-pass}

Since the variables in the optimization formulation cannot be computed offline, the strategy used in online aggregation is to execute the bi-level sampling process and check the constraint at runtime. However, this does not simplify the problem since we have to decide what is better: sample more chunks, i.e., increase $n$? or sample more tuples from the current chunk, i.e., increase $m_{j}$? While the variance decreases in both cases, we cannot quantify the overall impact on the objective function.

The solution we propose is to set the number of sampled chunks $n$ and solve the optimization formulation only with $m_{j}$ variables. The objective function reduces to minimizing the extraction time $T_{\textit{CPU}}(\overline{m_{j}})$, i.e., the number of extracted tuples. We set $n=N$ because this eliminates the term corresponding to the variance between chunks in Eq.~\eqref{eq:variance-est} while preserving the bi-level nature of the sampling process. Even with this simplification, we cannot derive closed-form formulae for the $m_{j}$ variables without having knowledge of the chunk heterogeneity index. In this situation, our strategy is to decompose the optimization formulation into separate problems for each chunk. These problems have the potential to become independent because of the simplification applied to the variance. For this to be the case, though, we have to identify $\textit{Var}_{\textit{max}}^{j}$ values for each chunk that guarantee the global constraint is satisfied if the constraints at chunk-level are satisfied. Since the same sampling without replacement process is performed inside a chunk, the local constraints can be written as:
\begin{equation}\label{eq:local-chunk-constraint}
\frac{M_{j}} {m_{j}} \frac{M_{j}-m_{j}} {m_{j}-1} \sum_{i\in C'_{j}} {\left( x_{i} - \frac{y'_{j}} {m_{j}} \right)^{2} } \leq \textit{Var}_{\textit{max}}^{j}
\end{equation}
We obtain the relationship $\sum_{j=1}^{N} {\textit{Var}_{\textit{max}}^{j}} \leq \textit{Var}_{\textit{max}}$ between $\textit{Var}_{\textit{max}}^{j}$ and $\textit{Var}_{\textit{max}}$ by summing up the local constraints and enforcing the global constraint. While it is possible to choose the values of $\textit{Var}_{\textit{max}}^{j}$ based on chunk properties, a simpler solution is to equally divide $\textit{Var}_{\textit{max}}$ across the $N$ chunks, i.e., $\textit{Var}_{\textit{max}}^{j} = \frac{1}{N} \textit{Var}_{\textit{max}}$, $\forall j$. This constant solution guarantees that more tuples are sampled from chunks with higher variability and less from chunks that are homogeneous. Since the input to our problem is $\epsilon$ rather than $\textit{Var}_{\textit{max}}$, we determine the corresponding values $\epsilon^{j}$ that satisfy the modified optimization formulation. They are given in Theorem~\ref{thm:optimal-epsilon}.

\begin{thm}\label{thm:optimal-epsilon}
Given an aggregate query and an accuracy $\epsilon$, a bi-level sampling procedure in which the estimate of each chunk aggregate satisfies accuracy $\epsilon^{j}=\epsilon$, $\forall j$, produces a global estimate that satisfies accuracy $\epsilon$ when all the chunks are sampled.
\end{thm}
\begin{proof}
The estimator $\widehat{\tau}$, its variance $\Var{\widehat{\tau}}$, and the accuracy $\epsilon$ are connected by the following equation:
\begin{equation}\label{eq:proof:1:1}
\frac{F(\Var{\widehat{\tau}})}{\widehat{\tau}} \leq \epsilon \Longleftrightarrow F(\Var{\widehat{\tau}}) \leq \epsilon \cdot \widehat{\tau}
\end{equation}
where $F$ is a function of the variance, e.g., the inverse of the cumulative normal distribution function. This is what we have to prove according to the independent sampling processes executed for every chunk. We know that the same equation holds for each chunk-level estimator, i.e., $F\left(\Var{\widehat{\tau^{j}}}\right) \leq \epsilon^{j} \cdot \widehat{\tau^{j}}$, but with its corresponding accuracy $\epsilon^{j}$. When all the chunks are included in the sample, bi-level sampling degenerates into stratified sampling. In this case, the following equalities hold:
\begin{equation}\label{eq:proof:1:2}
\sum_{j=1}^{N} {\widehat{\tau^{j}}} = \widehat{\tau} \hspace*{0.5cm} \sum_{j=1}^{N} {\Var{\widehat{\tau^{j}}}} = \Var{\widehat{\tau}}
\end{equation}
The inequality in Eq.~\eqref{eq:proof:1:1} becomes $F\left(\sum_{j=1}^{N} {\Var{\widehat{\tau^{j}}}}\right) \leq \epsilon \cdot \sum_{j=1}^{N} {\widehat{\tau^{j}}}$. If we sum-up the inequalities across chunks, we obtain $\sum_{j=1}^{N} {F\left(\Var{\widehat{\tau^{j}}}\right)} \leq \sum_{j=1}^{N} {\epsilon^{j} \cdot \widehat{\tau^{j}}}$. This is guaranteed to be true because of the sampling procedure executed at each chunk---sampling stops only when the inequality is satisfied. If we set $\epsilon^{j}=\epsilon$, $\forall j$, we obtain $\sum_{j=1}^{N} {F\left(\Var{\widehat{\tau^{j}}}\right)} \leq \epsilon \cdot \sum_{j=1}^{N} {\widehat{\tau^{j}}}$. Thus, for the final inequality to hold, we have to prove that:
\begin{equation}\label{eq:proof:1:3}
F\left(\sum_{j=1}^{N} {\Var{\widehat{\tau^{j}}}}\right) \leq \sum_{j=1}^{N} {F\left(\Var{\widehat{\tau^{j}}}\right)}
\end{equation}
Since we use the inverse of the normal cumulative distribution function as function $F$ to derive confidence intervals in bi-level sampling, this inequality holds as long as $\epsilon^{j}=\epsilon$, $\forall j$. This proves the theorem.
\end{proof}

Theorem~\ref{thm:optimal-epsilon} provides an algorithm to independently sample across chunks, thus allowing for maximum concurrency. It guarantees that -- in the worst scenario when all the chunks are sampled -- the aggregate estimate meets the required accuracy. It is important to emphasize that an estimate can be computed at any instant during the process and the accuracy can be met before all the chunks are considered. This can be done exclusively from the raw data---without knowledge of any other statistical information beyond simple chunk-level tuple counts. Moreover, this worst-case algorithm provides improvements over parallel external tables -- the standard mechanism for raw data processing -- in terms of the number of tuples extracted.

%%%%%%%%%%%%%%%%%%%%%%%%%%%%%%%%%%%%%%%%%%%%%%%%%%%%%%%%%
\subsection{Resource-Aware Sampling}\label{sec:ola-raw-sampling:resource}

The single-pass sampling strategy guarantees that the required accuracy is met after a pass over the data. While theoretically sound for a worst-case scenario, it does not take into account the runtime conditions. As soon as the local accuracy is satisfied, sampling for the chunk is terminated, even though computational resources may be available. As a result, the opportunity to further reduce the total variance of the estimator is missed. We address this shortcoming by introducing a novel bi-level sampling strategy that dynamically and adaptively determines how much to sample from a chunk by continuously monitoring the resource utilization of the system. While the original goal of single-pass sampling is preserved, we aim to maximize the chunk-level sample size without decreasing the rate at which chunks are processed.

\begin{minipage}{.35\textwidth}
In resource-aware sampling, the decision of when to stop sampling from a chunk is based on the availability of system resources -- I/O bandwidth and CPU threads -- in addition to the chunk accuracy. This requires continuous system monitoring at runtime. We monitor the buffer into which chunks are stored before \texttt{EXTRACT} and the pool of threads available to \texttt{EXTRACT}, each time an estimate is generated from a chunk, i.e., $t^{\textit{eval}}$. \end{minipage}\hfill
%%%%%%%%%%%
\begin{minipage}{.6\textwidth}
\begin{figure}[H]
\begin{center}
	\includegraphics[width=\textwidth]{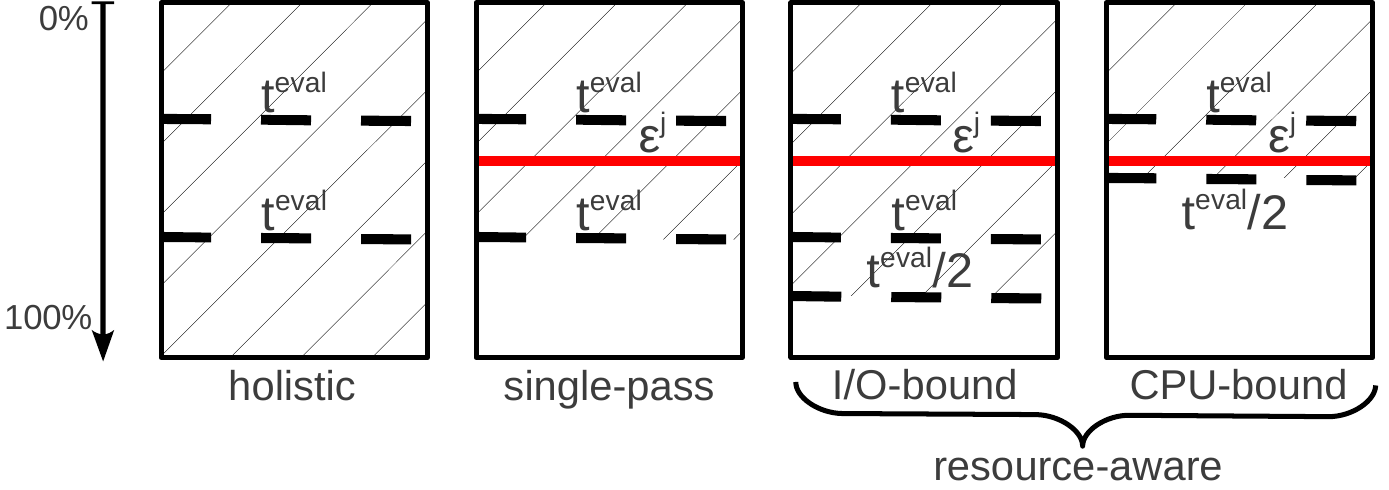}
	\caption{Bi-level sampling strategies in OLA-RAW.}
	\label{fig:res-threshold}
\end{center}
\end{figure}
\end{minipage}\hfill
%%%%%%%%%%%
The cost of monitoring does not incur significant overhead and it allows for the early detection of changes in the workload. As long as the number of threads is larger than the number of chunks in the buffer, processing is I/O-bound. Otherwise, it is CPU-bound. Since the objective function is the maximum of the two, we take different decisions in each case. They are depicted in Figure~\ref{fig:res-threshold}. In the case of I/O-bound in-situ processing, $t^{\textit{eval}}$ is halved only after the local chunk accuracy $\epsilon^{j}$ is achieved. For CPU-bound processing, this is done immediately after the first chunk estimate is generated. The rationale for these decisions is as follows. For I/O-bound processing, our goal is to end a chunk as soon as another chunk becomes available and there are no threads in the pool. In the case of CPU-bound tasks, the goal is to finalize a chunk as soon as the accuracy is met. By decreasing $t^{\textit{eval}}$, we increase the monitoring frequency in the hope that we detect the triggering factor -- accuracy or resource utilization -- as early as possible.

The timing parameter $t^{\textit{eval}}$ plays a central role in the resource-aware parallel sampling procedure proposed in this paper. Its value controls the frequency at which chunk-level estimates are produced and the monitoring interval. Moreover, $t^{\textit{eval}}$ is shared across the \texttt{EXTRACT} threads in order to avoid the inspection paradox. Thus, determining the appropriate values for $t^{\textit{eval}}$ is of major importance. When processing starts, we assign $t^{\textit{eval}}$ a lower bound value, e.g., $1 \textit{ms}$. $t^{\textit{eval}}$ cannot decrease below this value at any time. The first several chunks are extracted using the initial lower bound value because CPU resources are plentiful---it takes some time until the \texttt{EXTRACT} pipeline is filled. During this calibration process, our goal is to accurately determine the $t^{\textit{eval}}$ value at which the chunk accuracy is achieved. We set $t^{\textit{eval}}$ to the calibration average in order to reduce the number of monitoring operations. The calibration continues during the entire processing and $t^{\textit{eval}}$ is updated with the current running average. A clear upper bound for $t^{\textit{eval}}$ is $\delta$---the time interval at which results have to be presented to the user. Values larger than $\delta$ introduce hiccups in estimation. Another upper bound is represented by the time to process an entire chunk---this time is also monitored. In this case, the process degenerates to chunk-level sampling. Thus, the minimum of the two is taken as the upper bound. The monitoring in resource-constrained situations is done using an exponential decay scheme. The value of $t^{\textit{eval}}$ is repeatedly halved to increase the monitoring frequency. Nonetheless, the lower bound is enforced. This scheme is especially important when the accuracy $\epsilon$ is always satisfied at the first monitoring and the processing is CPU-bound since it forces the global $t^{\textit{eval}}$ to drop.

%%%%%%%%%%%%%%%%%%%%%%%%%%%%%%%%%%%%%%%%%%%%%%%%%%%%%%%%%
\subsection{Summary}\label{sec:ola-raw-sampling:summary}

In this section, we introduce three parallel sampling methods for online aggregation over raw data derived from bi-level sampling. As far as we know, this is the first time when bi-level sampling is used for estimation in online aggregation. The three methods are depicted in Figure~\ref{fig:res-threshold}. Holistic sampling is a direct realization of bi-level sampling that avoids the inspection paradox by generating chunk estimates at constant time intervals. Single-pass sampling improves upon the holistic scheme by reducing the size of the sample inside a chunk, while guaranteeing that the required accuracy is achieved in a single pass over the data. We provide an analytical scheme to determine the chunk-level accuracy $\epsilon^{j}$. Single-pass sampling is more efficient when processing is CPU-bound because it reduces the number of extracted tuples. In I/O-bound scenarios, holistic sampling is more efficient because it accesses fewer chunks. Resource-aware sampling is optimal independent of the processing since it dynamically and adaptively determines the appropriate chunk sample size and adapts to the characteristics of the workload. While it degenerates to holistic sampling in I/O-bound scenarios and to single-pass sampling for CPU-bound workloads, resource-aware sampling identifies the changes in processing and adapts accordingly by continuously monitoring and updating the estimation time interval $t^{\textit{eval}}$. It improves upon single-pass sampling by further reducing the chunk sample size or the estimator variance, while preserving the streaming access pattern.

%%%%%%%%%%%%%%%%%%%%%%%%%%%%%%%%%%%%%%%%%%%%%%%%%%%%%%%%
%\input{sample-maintenance}
\section{BI-LEVEL SAMPLE SYNOPSIS}\label{sec:synopsis}

Extracting a bi-level sample from raw data is expensive. Since data exploration requires a sequence of queries, this process becomes the bottleneck if executed from scratch for each query. It is important to notice that extracting the sample once and using it for multiple queries is not possible in our scenario because sampling is driven by the current query. It is very likely that the sample corresponding to a query cannot be used for another query, e.g., it does not include all the required columns. In this case, access to the raw data is necessary in order to extract the missing columns. This requires complete resampling and provides little opportunity for improvement. The more interesting case is when the same sample can be used across queries, however, the required accuracy for the new query cannot be met. This can happen because, for example, the user-specified accuracy $\epsilon$ increases, the new query is more selective, or the new aggregate is different. In this situation, the goal is to preserve and incrementally maintain the sample.

We propose a novel \textit{memory-resident bi-level sample synopsis} that is built and maintained based on the query workload. Since the size of the synopsis is bound by a specified memory budget $B$, the synopsis is not necessarily identical to the bi-level sample extracted from the raw data and used for estimation. The synopsis is built based on a query specified by the user, i.e., origin query, and is updated by each subsequent query---only if the query cannot be answered exclusively using the synopsis. A query that cannot take advantage of the synopsis triggers a complete rebuild automatically.

Several aspects underlie the novelty of the bi-level sample synopsis with respect to other sampling strategies used in online aggregation. First, the proposed synopsis is computed entirely at runtime and does not require any preprocessing, e.g., offline shuffling or sampling. Second, the synopsis is driven by the queries. There is no generic process to generate universal samples applicable to any query. Third, the synopsis caches in memory a subsample of the data to be used directly by subsequent queries. No other solution considers a query sequence. Each query starts sampling the data from scratch. While the idea of caching extracted data in memory is similar to NoDB~\cite{nodb}, the proposed bi-level sample synopsis does not require caching full columns. This is the only case in which NoDB can avoid accessing the raw data for subsequent queries. 

%%%%%%%%%%%
\begin{figure}[htbp]
	\centering
	\includegraphics[width=0.47\textwidth]{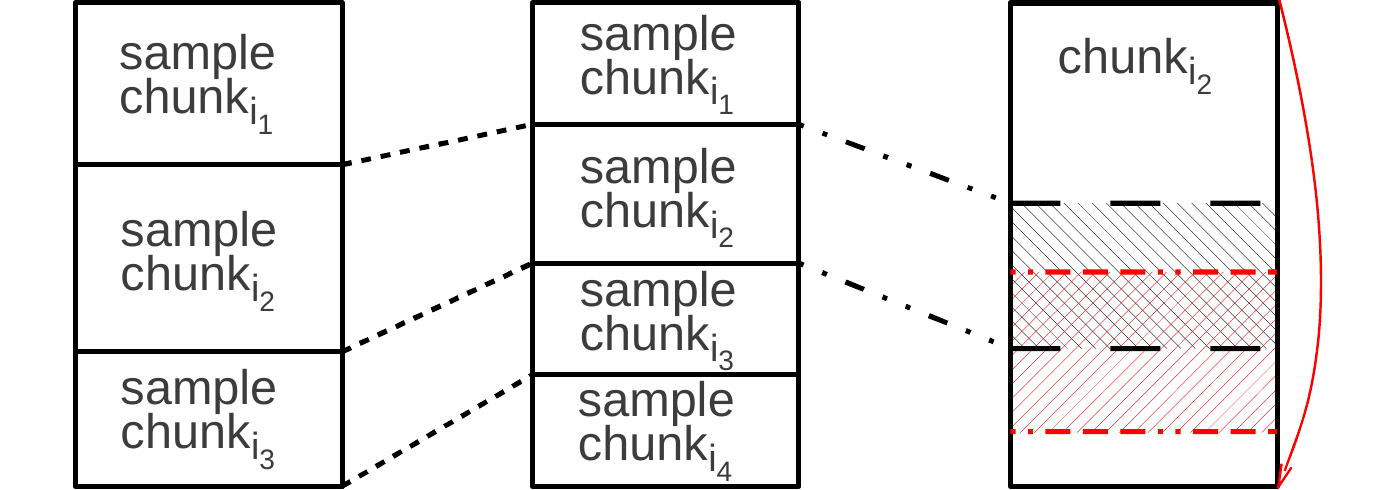}
	\caption{Maintenance of bi-level sample synopsis under new chunk insertion and existing chunk resampling.}
	\label{fig:sample-maintain}
\end{figure}
%%%%%%%%%%%

%%%%%%%%%%%
\subsection{Synopsis Construction}\label{ssec:synopsis:construction}

Our goal is to build a bi-level sampling synopsis with memory budget $B$ out of the samples extracted from raw data used in the estimation. Remember that the samples used in estimation are driven by the accuracy $\epsilon$. Thus, the total size can exceed memory budget $B$. The standard practice in online aggregation is to compute the quantities necessary for estimation -- $m_{j}$, $y'_{j}$, and $y''_{j}$ in our case -- and then discard the sample. Since sampling from raw data is expensive due to extraction, we preserve (some of) the samples in the bi-level synopsis.

We build the bi-level sample synopsis following a process similar to reservoir sampling~\cite{reservoir}. The chunks are considered for insertion in the random order they are extracted for estimation. As long as the memory budget $B$ is not exhausted, all the tuples extracted from a chunk are added to the synopsis---organized based on chunks. Complications appear when the memory budget is filled. In this case, we propose a \textit{variance-driven insertion strategy}---depicted in Figure~\ref{fig:sample-maintain}. The budget $B$ is divided across the chunks in the synopsis and the new one proportionally to their local variance for the current query. The larger the internal variance, the more synopsis space a chunk gets. As shown in Figure~\ref{fig:sample-maintain}, this requires dropping some of the sampled tuples inside the chunk while preserving a sample with smaller size. We realize this by discarding the tuples at the front of the random permutation based on which tuples are extracted. By following the variance-driven insertion strategy, the synopsis is guaranteed to contain a bi-level sample at any time instant. This is important because processing can finish at any moment. In the extreme case when all the chunks are included, the synopsis degenerates into a stratified sample. For this type of sampling, we can determine the optimal order in which to extract chunks in a subsequent query---the decreasing order of the chunk variance.

%%%%%%%%%%%
\subsection{Synopsis Maintenance}\label{ssec:synopsis:maintenance}

Whenever a chunk is contained in the synopsis, we use it for estimation. If more tuples are required, we extract them starting at the end of the samples stored in the synopsis using the random permutation corresponding to the chunk. When we reach the end of the permutation, we restart from the beginning in a circular random scan (Figure~\ref{fig:sample-maintain}). The new samples have to be merged with the existing ones already in the synopsis from the same chunk while satisfying the budget constraint. We have to determine the size of the merged sample and the content. The size is computed based on the same proportional variance allocation. We include in the synopsis tuples starting from the end of the new samples until we fill the allocated space. Tuples already in the synopsis are kept only if there is enough space. This strategy -- depicted on the right side of Figure~\ref{fig:sample-maintain} -- preserves the incremental sample generation while also refreshing the sample continuously across queries. 

%%%%%%%%%%%
\subsection{Estimation}\label{ssec:synopsis:estimation}

The benefit of the memory-resident synopsis is that it allows for estimates to be computed faster and -- in the best case scenario -- without even accessing the raw data. Since the synopsis is a bi-level sample, the formulae for the estimator and variance given in Section~\ref{sec:bilevel-sampling:estimation} can be applied directly. Moreover, the resource-aware sampling mechanism works without changes. The only difference is that chunks are accessed from the synopsis rather than from secondary storage.

We differentiate between two cases that provide separate optimization opportunities. If not all the chunks are represented in the synopsis, we have to enforce the original chunk order used to build the synopsis---both for the in-memory and non-in-memory chunks. Since I/O and estimation are overlapped, the reading of non-in-memory chunks is triggered concurrently with the evaluation of the in-memory chunks. An interesting situation arises when the local chunk accuracy is not satisfied for a chunk in the synopsis. In this case, the chunk may require additional sampling in order to meet the user specified accuracy $\epsilon$. Instead of scheduling the chunk for reading immediately, we append it at the end of the processing queue. We give priority to new chunks because they have infinite variance and, otherwise, we may introduce unnecessary delays in estimation. Notice, though, that the accuracy $\epsilon$ is still guaranteed to be achieved in a single pass over the raw data. In the second case, the synopsis contains samples from all the chunks. Once we set the access order, we start processing chunks from the synopsis. Whenever the local accuracy is not met, we schedule the chunk for reading. As an optimization, we can change the reading order based on the chunk variance---higher variance chunks are scheduled first. This is possible because the synopsis is a stratified sample.

%%%%%%%%%%%%%%%%%%%%%%%%%%%%%%%%%%%%%%%%%%%%%%%%%%%
\begin{comment}
%Where should we store the samples.
After consumed by the execution engine, sample is preserved in the memory. The order of sample chunks in the memory is the same as their generation order, which means sequential process these sample chunks equals to regenerate the bi-level sample.
%How to refresh the samples.
During queries processing, sample is continuously appended to the memory. However, the memory budget is limited, it is likely that bi-level sample cannot fits fully in memory. When it happens, an efficient swap mechanism has to be designed. The goal of sample maintenance is to quickly generate the estimation for the following queries. We propose a swap strategy 

%How to use samples for the following queries.
For the following queries, the system could immediately generates an estimated result based on the in-memory sample. If the accuracy of estimated result is good enough, the query execution immediately finishes without accessing any data from raw file, which is an optimal situation. Otherwise, system has to produce more new samples from raw file to improve the accuracy of estimated result. 
\end{comment}
%%%%%%%%%%%%%%%%%%%%%%%%%%%%%%%%%%%%%%%%%%%%%%%%%%%

%%%%%%%%%%%%%%%%%%%%%%%%%%%%%%%%%%%%%%%%%%%%%%%%%%%%%%%%
%\input{experiments}
\section{EXPERIMENTAL EVALUATION}\label{sec:experiments}

The objective of the experimental evaluation is to investigate the performance of OLA-RAW bi-level sampling across a variety of datasets -- synthetic and real -- and workloads---including a single query as well as a sequence of queries. Additionally, the difference among the proposed parallel bi-level sampling strategies -- holistic, single-pass, and resource-aware -- and with respect to chunk-level sampling is throughly analyzed. Specifically, the experiments we design are targeted to answer the following questions:

\begin{compactitem}
\item How does OLA-RAW bi-level sampling compare with external tables and chunk-level sampling as a function of the degree of parallelism?
\item How much data -- chunks and tuples -- are processed in order to answer a query with specified accuracy?
\item Is there any palpable difference -- in accuracy and resource utilization -- among the proposed bi-level sampling strategies? What about chunk-level sampling?
\item What is the effect of the bi-level sample synopsis on the execution of a sequence of queries?
\item Are the OLA-RAW confidence bounds correct? Is the inspection paradox an issue in practice?
\end{compactitem}

%%%%%%%%%%%%%%%%%%%%%%%%%%%%%%%%%%%%
\subsection{Setup}\label{ssec:experiments:setup}

%%%%%%%%%%%%%%%%%%%%%%%%%%%%%%%%%%%%
\textbf{Implementation.}
We implement OLA-RAW based on the super-scalar pipeline architecture of SCANRAW~\cite{scanraw,scanraw:tods}. Speculative loading is disabled and replaced with NoDB-style~\cite{nodb} caching. However, only the sample synopsis is cached in memory. In the experiments, we use CSV and FITS file formats. While SCANRAW extractors are already available, they are not immediately applicable to sampling which requires direct access to random chunks. Moreover, the tuples inside a chunk have to be accessed in random order and they have to be extracted incrementally. We introduce an estimation controller that manages the $t^{\textit{eval}}$ timing parameter and coordinates the \texttt{EXTRACT} instances. Whenever $t^{\textit{eval}}$ corresponding to an \texttt{EXTRACT} expires, the controller interrupts the instance and pulls its sample for estimation. Based on the stopping criteria and the resource utilization of the system, a decision is made on continuing the \texttt{EXTRACT} and its corresponding $t^{\textit{eval}}$ timer. The estimation controller produces an estimate and confidence bounds for the user at $\delta$ time intervals. The sample synopsis is also maintained by the controller which executes the maintenance procedure for each incoming chunk sample.

%%%%%%%%%%%%%%%%%%%%%%%%%%%%%%%%%%%%
\textbf{System.}
All experiments are executed on a standard server with 2 AMD Opteron 6128 series 8-core processors (64 bit) -- 16 cores -- 64 GB of memory, and four 2 TB 7200 RPM SAS hard-drives configured RAID-0 in software. Each processor has 12 MB L3 cache while each core has 128 KB L1 and 512 KB L2 local caches. \eat{According to \texttt{hdparm}, }The cached and buffered read rates are 3 GB/second and 565 MB/second, respectively. If the experiment consists of a single query, we always enforce data to be read from disk by cleaning the file system buffers before execution. In experiments over a sequence of queries, the buffers are cleaned only before the first query. Thus, the second and subsequent queries can access cached data. Ubuntu 14.04.2 SMP $64$-bit with Linux kernel 3.13.0-43 is the operating system.

%%%%%%%%%%%%%%%%%%%%%%%%%%%%%%%%%%%%
\textbf{Data.}
We run experiments over four datasets---one synthetic and three real. Three of the datasets are in CSV format, while \texttt{ptf-fits} is the original PTF dataset in the FITS binary format. Table~\ref{tbl:datasets} depicts the characteristics of the datasets. \texttt{ptf-csv} and \texttt{ptf-fits} contain the same data in text and binary format, thus the difference in size. The PTF dataset contains 1 billion transient detections, i.e., tuples. Each tuple has 8 attributes, i.e., columns, 6 of which are real numbers with 10 decimal digits. \texttt{wiki} is extracted from the Wikipedia Traffic Statistics V2 dataset\footnote{\url{https://aws.amazon.com/datasets/4182}}. It contains the aggregated hits per Wikipedia page and the number of bytes transferred for each hour in March 2015. The size of this dataset is 19 GB in CSV format. The \texttt{synthetic} dataset contains randomly generated integers smaller than 1 billion grouped in tuples of 16 columns. Each column is generated using a different zipfian distribution ranging from uniform to extremely skewed. For all the datasets, the number of chunks is chosen such that the size of a chunk is in the order of tens to a hundred megabytes. This optimizes the disk throughput.

%%%%%%%%%%%%%%%%%%%%%%%%%%%%%%%%%%
\begin{table}[htbp]
  \begin{center}
    \begin{tabular}{l||r|r|r|r}

	\textbf{Dataset} & \textbf{\# Tuples} & \textbf{\# Chunks} & \textbf{\# Columns} & \textbf{Size} \\

	\hline
	
	\texttt{ptf-csv} & 1B & 1000 & 8 & 68 GB \\
	
	\texttt{ptf-fits} & 1B & 1000 & 8 & 60 GB \\

	\texttt{wiki} & 1.8B & 130 & 4 & 19 GB \\

	\texttt{synthetic} & 134M & 512 & 16 & 20 GB \\
    \end{tabular}
  \end{center}

\caption{Datasets used in the experiments.}\label{tbl:datasets}
\end{table}
%%%%%%%%%%%%%%%%%%%%%%%%%%%%%%%%%%

%%%%%%%%%%%%%%%%%%%%%%%%%%%%%%%%%%%%
\textbf{Methods.}
We perform experiments for three raw data processing methods. We use external tables (\texttt{EXT} in the figures) as a baseline for comparison. \texttt{EXT} computes the query result exactly by inspecting all the data in sequential order. There is no sampling or estimation, thus no overhead. The second method is parallel chunk-level sampling (\texttt{C} in the figures), while the third method is the proposed resource-aware OLA-RAW bi-level sampling (\texttt{BI} in the figures). This is the most advanced OLA-RAW sampling technique introduced in the paper. Each method is evaluated with a different number of worker threads allocated to \texttt{EXTRACT}. This number is specified after the method name, e.g., \texttt{BI-4} corresponds to OLA-RAW bi-level sampling with 4 threads.

%%%%%%%%%%%%%%%%%%%%%%%%%%%%%%%%%%%%
\textbf{Measurements.}
All our figures depict the error ratio as a function of time. The error ratio is defined as the relative confidence bounds width ($\frac{\textit{high} - \textit{low}}{\textit{estimate}}$) measured as the ratio between the difference of the confidence interval extremes and the estimate. In all the results, the error is computed at $95\%$ accuracy, i.e., $\epsilon=0.95$. The estimation time interval $\delta$ is set to 1 second. This generates estimates for the user at an almost continuous rate. The number of chunks and tuples extracted in chunk-level and bi-level sampling is recorded in order to identify the difference between the two methods. These quantities are presented as the relative ratio from the entire dataset---they are always $1$ for \texttt{EXT}.

%%%%%%%%%%%%%%%%%%%%%%%%%%%%%%%%%%%%
\subsection{Results}\label{ssec:experiments:results}

In this section, we present results for all the datasets and provide details on the queries and the experimental settings.

%%%%%%%%%%%%%%%%%%%%%%%%%%%%%%%%%%%%
\begin{figure*}[htbp]
	\begin{center}
		\includegraphics[width=0.9\textwidth,height=4.19in]{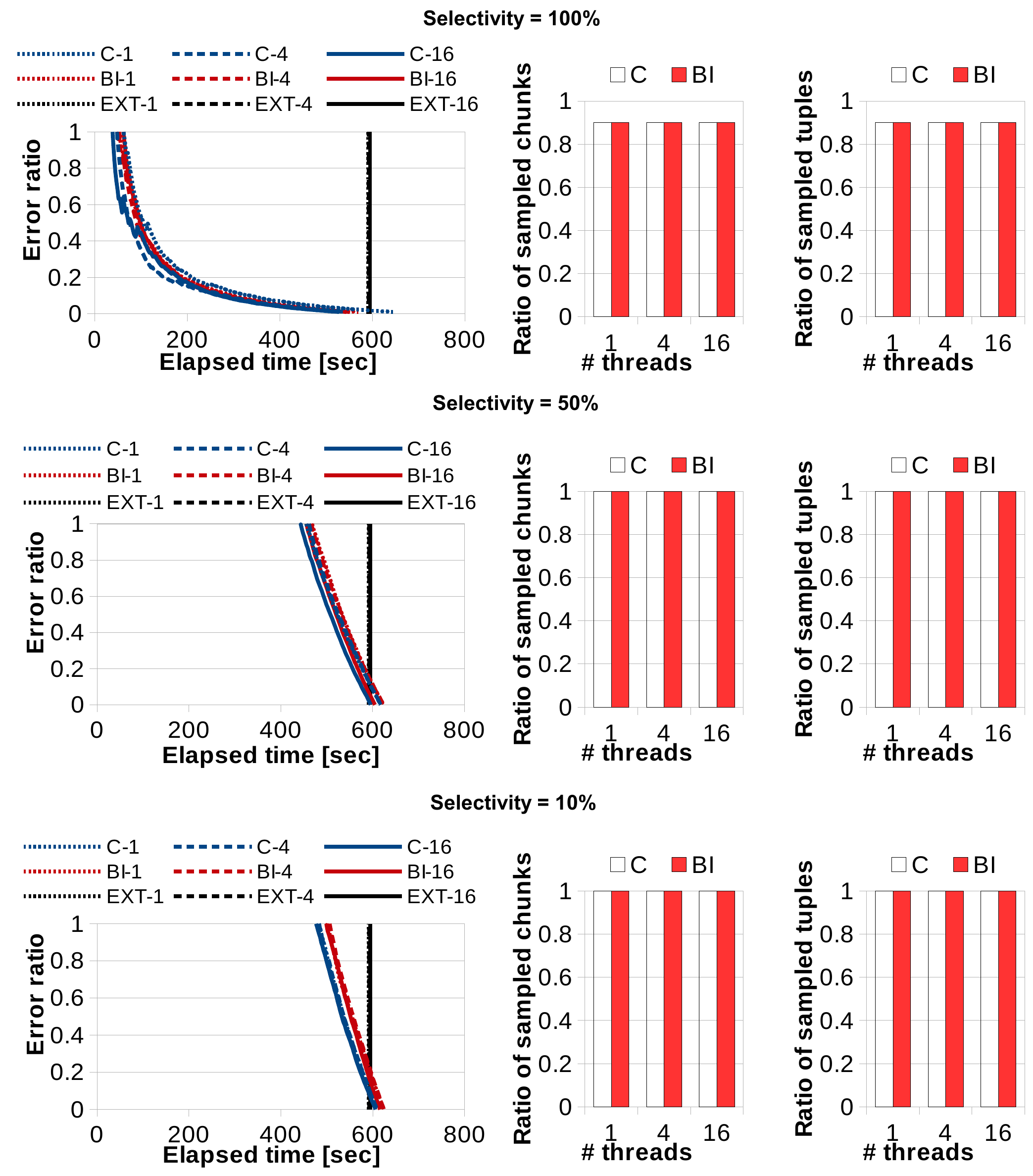}
		\caption{Results on the \texttt{ptf-fits} dataset.}
		\label{fig:ptf-fits}
	\end{center}
\end{figure*}
%%%%%%%%%%%%%%%%%%%%%%%%%%%%%%%%%%%%%

%%%%%%%%%%%%%%%%%%%%%%%%%%%%%%%%%%%%
\subsubsection{Estimation Error and Convergence}\label{sssec:experiments:results:error}

%%%%%%%%%%%%%%%%%%%%%%%
\textbf{PTF-FITS.}
The \texttt{ptf-fits} dataset contains the PTF project data in the original storage format. The content is identical to \texttt{ptf-csv} and consists of 1 billion tuples. Each tuple has 8 attributes, 6 of which are real numbers with as many as 10 decimal digits. The increased precision is necessary for accurately representing the celestial coordinates and the time of the detected objects. The detections are sorted according to their detection time. This creates clumps of tuples clustered both on time and position---each night there are only a handful of clumps. Moreover, these clumps are not uniformly distributed in the sky. They are skewed around the location of the telescope. The aggregate query sums up a linear expression of the six real number attributes in the dataset---it is the same as for \texttt{ptf-fits}. The selectivity is controlled by range predicates on the domain of the positional attributes. Due to skewness, the number of tuples in the result does not immediately correspond to the selectivity value. The accuracy is set to $95\%$.

%%%%%%%%%%%%%%%%%%%%%%%%%%%%%%%%%%%%%%%%%%%%%%%%
\begin{minipage}{.5\textwidth}
\begin{sql}
SELECT SUM($\sum_{i}{c_{i}\cdot A_{i}}$)
FROM ptf-fits
WHERE $A_{1}$ BETWEEN $[l_{1}, r_{1}]$ AND $A_{2}$ BETWEEN $[l_{2}, r_{2}]$
\end{sql}
\end{minipage}\hfill
%%%%%%%%%%%%%%%%%%%%%%%%%%%%%%%%%%%%%%%%%%%%%%%%
\begin{minipage}{.4\textwidth}
The results are depicted in Figure~\ref{fig:ptf-fits}. Since FITS is a binary format, processing is I/O-bound independent of the number of threads.
\end{minipage}\hfill
This is the reason the curves for all the methods collapse into a single curve---there are minor differences due to the experimental conditions. In I/O-bound scenarios, \texttt{BI} degenerates to \texttt{C} since there are enough CPU resources to process all the tuples in a chunk. However, since estimates are computed more frequently, \texttt{BI} incurs additional overhead, thus, the small delay compared to \texttt{C}. The delay with respect to \texttt{EXT} for low selectivity is due to the random order in which chunks have to be processed for sampling. The main takeaway from these figures is that meaningful estimations can be produced in a timely manner only for high selectivities. At lower selectivity, accurate estimations are produced only close to the time when \texttt{EXT} finishes execution. We emphasize the cause for these results is the skewness and periodicity of the data. Even for $100\%$ selectivity, $90\%$ of the chunks are processed to reach the required accuracy.

%%%%%%%%%%%%%%%%%%%%%%%%%%%%%%%%%%%%
\begin{figure*}[htbp]
	\begin{center}
		\includegraphics[width=0.9\textwidth,height=4.19in]{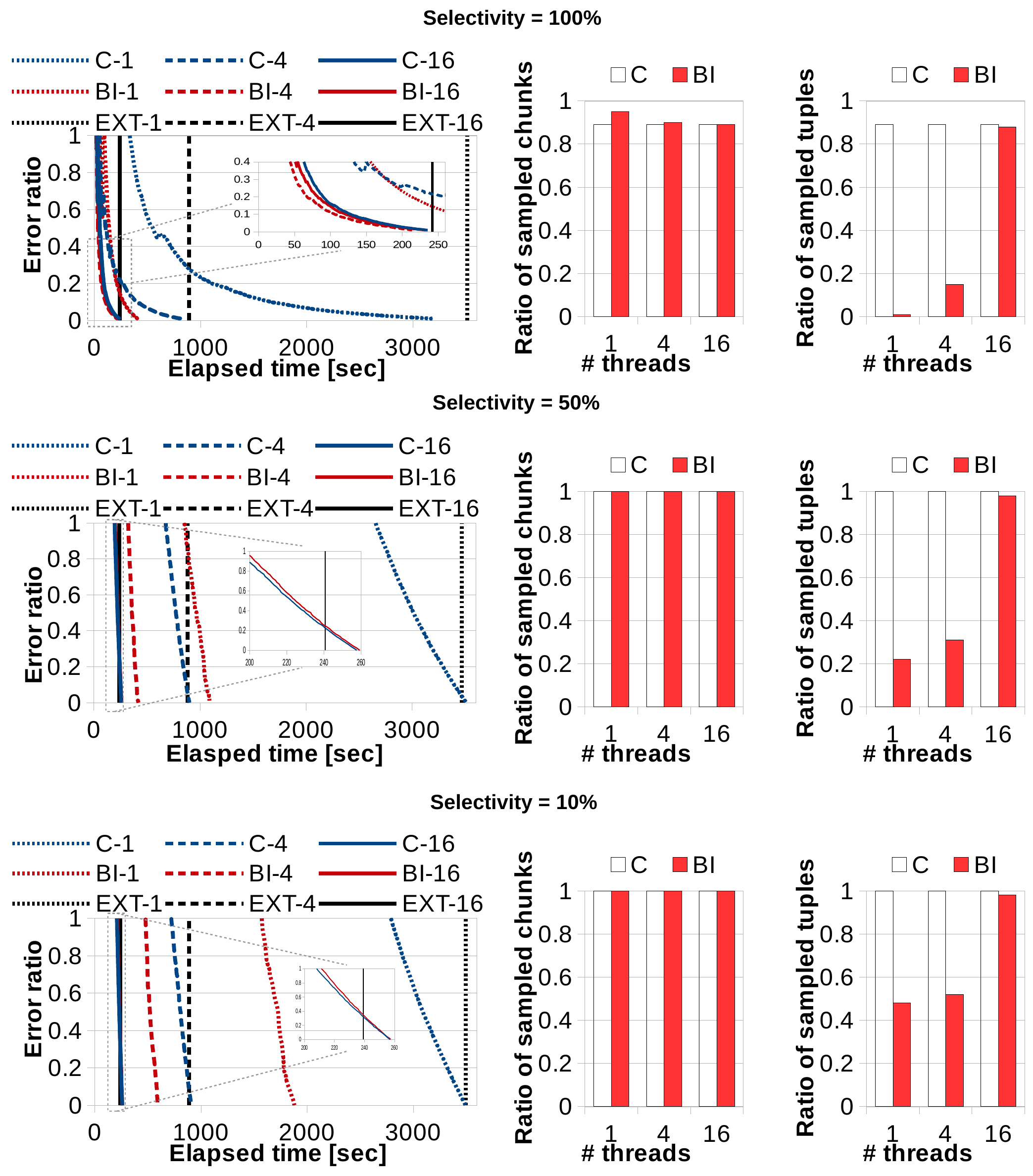}
		\caption{Results on the \texttt{ptf-csv} dataset.}
		\label{fig:ptf-txt}
	\end{center}
\end{figure*}
%%%%%%%%%%%%%%%%%%%%%%%%%%%%%%%%%%%%%

\textbf{PTF-CSV.}
Figure~\ref{fig:ptf-txt} depicts a complete comparison between the three methods considered in the experiments for the \texttt{ptf-csv} dataset. The query used in the evaluation sums up a linear expression of the six real number attributes in the dataset. The number of tuples included in the aggregate is controlled by a selection predicate on two of the six attributes. Selectivity $x\%$ corresponds to a range predicate that covers $x\%$ of the predicate attribute domain. Since \texttt{EXT} computes the exact result, its error is infinite until the computation finishes, at which point it becomes zero. As the number of threads increases, the execution time for \texttt{EXT} decreases linearly. This shows that processing the \texttt{ptf-csv} dataset is CPU-bound---\texttt{EXTRACT} is the bottleneck. For 1 and 4 threads, the \texttt{BI} error is always lower than the \texttt{C} error at a given time instant. The reason for this is the number of tuples processed by the two methods. Although the two methods process almost the same number of chunks, it is the number of processed tuples that determines the execution time in this CPU-bound scenario. While \texttt{C} has to process all the tuples inside each chunk, \texttt{BI} stops as soon as the required accuracy is satisfied. In this case, this happens after a small number of tuples which shows that chunks contain homogeneous tuples while they are highly different among themselves. This makes sense since candidates are added to the PTF catalog based on their detection time. For 16 threads, two interesting phenomena occur. First, \texttt{EXT} almost catches up with the sampling methods. This is a clear sign that processing becomes I/O-bound. Second, \texttt{BI} reduces to \texttt{C} since there are enough CPUs to process all the tuples inside a chunk. As expected, low selectivity has a negative impact on estimation. Since fewer tuples are part of the estimator, it takes longer to satisfy the required accuracy. In the worst case, all the chunks and all the tuples have to be processed to achieve the required accuracy and sampling does not bring any benefit over external tables. This happens for \texttt{C} for all the settings when selectivity is $50\%$ and $10\%$. Although \texttt{BI} also inspects all the chunks, it does so much faster since it does not extract all the tuples---the number of inspected tuples increases with the selectivity decrease, thus the increase in execution time. For 16 threads and low selectivity, external tables are the best alternative since sampling incurs overhead in data access.

%%%%%%%%%%%%%%%%%%%%%%%%%%%%%%%%%%%%
\begin{figure*}[htbp]
	\begin{center}
		\includegraphics[width=0.9\textwidth,height=4.19in]{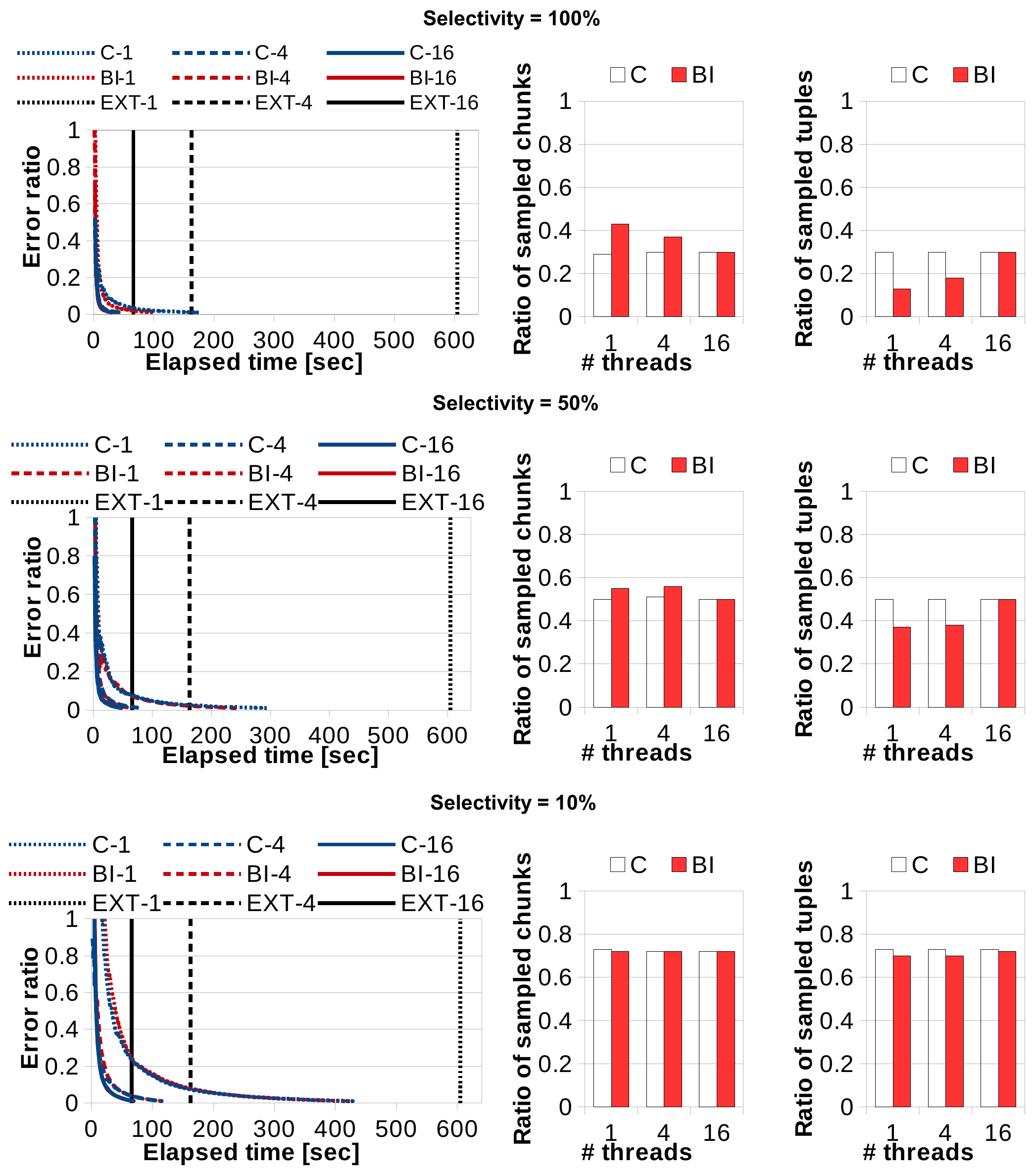}
		\caption{Results on the \texttt{synthetic} dataset.}
		\label{fig:synthetic}
	\end{center}
\end{figure*}
%%%%%%%%%%%%%%%%%%%%%%%%%%%%%%%%%%%%%

%%%%%%%%%%%%%%%%%%%%%%%
\textbf{Synthetic.}
The \texttt{synthetic} dataset contains randomly generated integers smaller than 1 billion grouped in tuples of 16 columns. Each column is generated using a different zipfian distribution. The parameter of the zipfian distribution is set at equal intervals across the range $[0, 4)$, i.e., $0, 0.25, 0.5, \dots, 3.75$, which gives columns having distributions ranging from uniform to extremely skewed. This allows for a thorough sampling evaluation. The query sums up a linear expression of the 16 attributes while applying the selection predicate to the uniformly distributed attribute $A_{1}$. This guarantees that the number of tuples in the result is correlated with the selectivity. The accuracy is set to $95\%$ as for the other datasets.

%%%%%%%%%%%%%%%%%%%%%%%%%%%%%%%%%%%%%%%%%%%%%%%%
\begin{minipage}{.35\textwidth}
\begin{sql}
SELECT SUM($\sum_{i}{c_{i}\cdot A_{i}}$)
FROM synthetic
WHERE $A_{1} < \frac{\text{selectivity}}{100} \cdot 10^{9}$
\end{sql}
\end{minipage}\hfill
%%%%%%%%%%%%%%%%%%%%%%%%%%%%%%%%%%%%%%%%%%%%%%%%
\begin{minipage}{.6\textwidth}
The results are depicted in Figure~\ref{fig:synthetic}. While they follow the same trend as for \texttt{ptf-csv}, there are several important differences. First, the number of chunks required for estimation is considerably lower. This is because the chunks are less heterogeneous.
\end{minipage}\hfill
The result is faster convergence, especially for the I/O-bound settings with $4$ and $16$ threads. Second, the difference between \texttt{BI} and \texttt{C} is less clear. While \texttt{BI} processes more chunks, \texttt{C} processes more tuples. However, due to data homogeneity inside chunks, both methods converge similarly because the time to read chunks in \texttt{BI} is matched by the time to process the additional tuples in \texttt{C}. Third, sampling methods are much faster than external tables across all the configurations. Again, this can be explained by the dependency of estimation on the data. \texttt{synthetic} is more amenable to estimation than the real datasets. This helps chunk-level sampling in particular.

%%%%%%%%%%%%%%%%%%%%%%%%%%%%%%%%%%%%
\begin{figure*}[htbp]
	\begin{center}
		\includegraphics[width=.95\textwidth,height=1.82in]{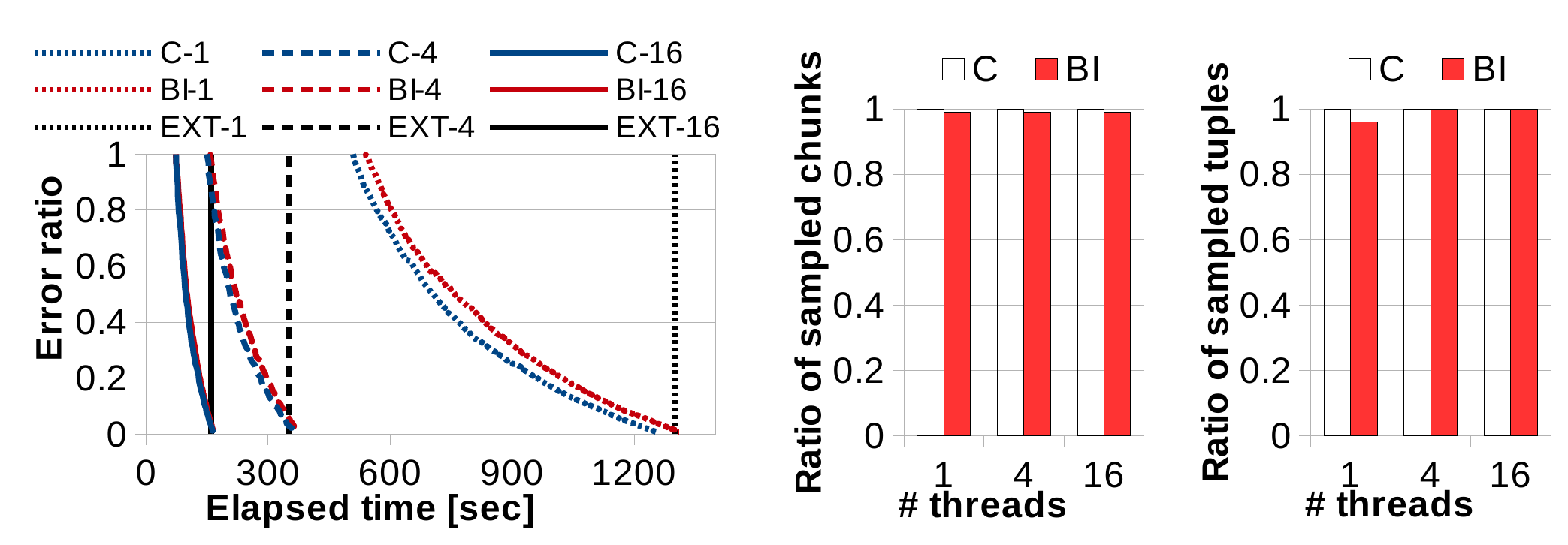}
		\caption{Results on the \texttt{wiki} dataset.}
		\label{fig:wiki}
	\end{center}
\end{figure*}
%%%%%%%%%%%%%%%%%%%%%%%%%%%%%%%%%%%%%

%%%%%%%%%%%%%%%%%%%%%%%
\textbf{Wiki.}
This experiment replicates the evaluation on the same dataset in~\cite{online-mapreduce}. Given the focus on a shared-memory architecture, we use a data sample having the size corresponding to a single node in~\cite{online-mapreduce}---one month instead of six months. The query counts the number of Wikipedia page hits on a per-language basis. \texttt{GROUP BY} is implemented by a separate query for each language. We present results for \texttt{en}---the results for other languages are similar.

%%%%%%%%%%%%%%%%%%%%%%%%%%%%%%%%%%%%%%%%%%%%%%%%
\begin{minipage}{.25\textwidth}
\begin{sql}
SELECT COUNT(hits)
FROM wiki
GROUP BY language
\end{sql}
\end{minipage}\hfill
%%%%%%%%%%%%%%%%%%%%%%%%%%%%%%%%%%%%%%%%%%%%%%%%
\begin{minipage}{.7\textwidth}
The results are depicted in Figure~\ref{fig:wiki}. The main point is that accurate estimation takes as long as the execution of \texttt{EXT} both for \texttt{BI} and \texttt{C}. This is explained by the large number of chunks that have to be inspected---almost $100\%$. \end{minipage}\hfill
The reason for such a large number of chunks is the low selectivity of the query. The number of tuples for a language is very small inside each chunk and this decreases the variance at a low rate. Parallelism improves the convergence rate significantly in this case---data are in CSV format.

%%%%%%%%%%%%%%%%%%%%%%%%%%%%%%%%%%%%
\subsubsection{Parallel Sampling Comparison}\label{sssec:experiments:results:sampling}

Figure~\ref{fig:comparison} depicts a comparison between the parallel bi-level sampling methods discussed in the paper -- holistic (\texttt{H}), single-pass (\texttt{S}), resource-aware (\texttt{BI}) -- and chunk-level sampling (\texttt{C}) over the \texttt{synthetic} dataset. The query used in the evaluation sums up a linear expression of the 16 attributes in the dataset, without any selectivity. Since our goal is to emphasize the difference between methods, we zoom the figures on the area where the difference is more accentuated. For CPU-bound settings -- 1 and 4 thread(s) -- \texttt{S} and \texttt{BI} are achieving the fastest error reduction. This is because they stop as soon as the required accuracy for a chunk is reached. \texttt{H} and \texttt{C} -- on the other hand -- do not stop until the complete chunk is extracted. \texttt{BI} incurs additional delay over \texttt{S} due to more frequent convergence checking. Although \texttt{H} produces estimates earlier, it is more sensitive to fluctuations in the chunk-level estimate, especially when there is no parallelism. When multiple threads are available, the step behavior of \texttt{C} can be clearly detected due to chunk processing inversion. A slow chunk delays estimation for all the subsequent faster chunks due to the inspection paradox. When an estimate can be finally computed, there is a steep decrease in the error. For I/O-bound settings with 16 threads, \texttt{BI} reduces to \texttt{C}, while \texttt{S} becomes the worst because it stops processing as soon as the required chunk accuracy is reached. \texttt{BI} identifies that sufficient CPU resources are available and continues to extract tuples. It also updates the timing parameter $t^{\textit{eval}}$ to reduce the estimation frequency. This is not the case for \texttt{H}, thus the larger delay. The main takeaway is that \texttt{BI} is always transforming into the best (or almost the best) strategy, independent of the resource constraints.

%%%%%%%%%%%%%%%%%%%%%%%%%%%%%%%%%%%%
\begin{figure*}[htbp]
	\begin{center}
		\subfloat[]{\includegraphics[width=2.07in,height=1.75in]{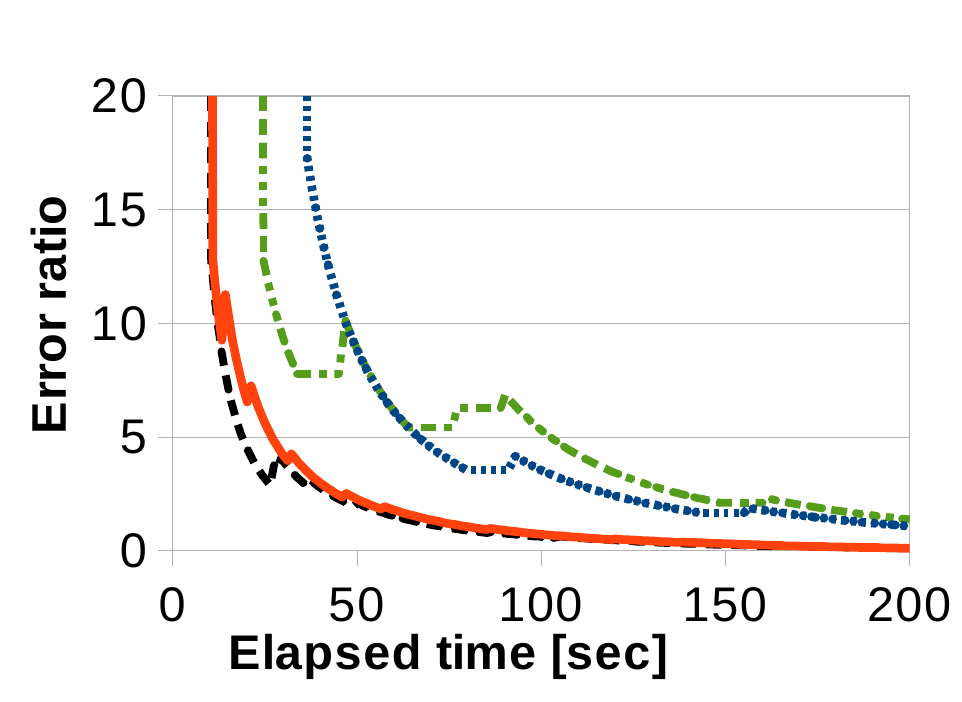}\label{fig:comparison-t1}}
		\subfloat[]{\includegraphics[width=2.07in,height=1.75in]{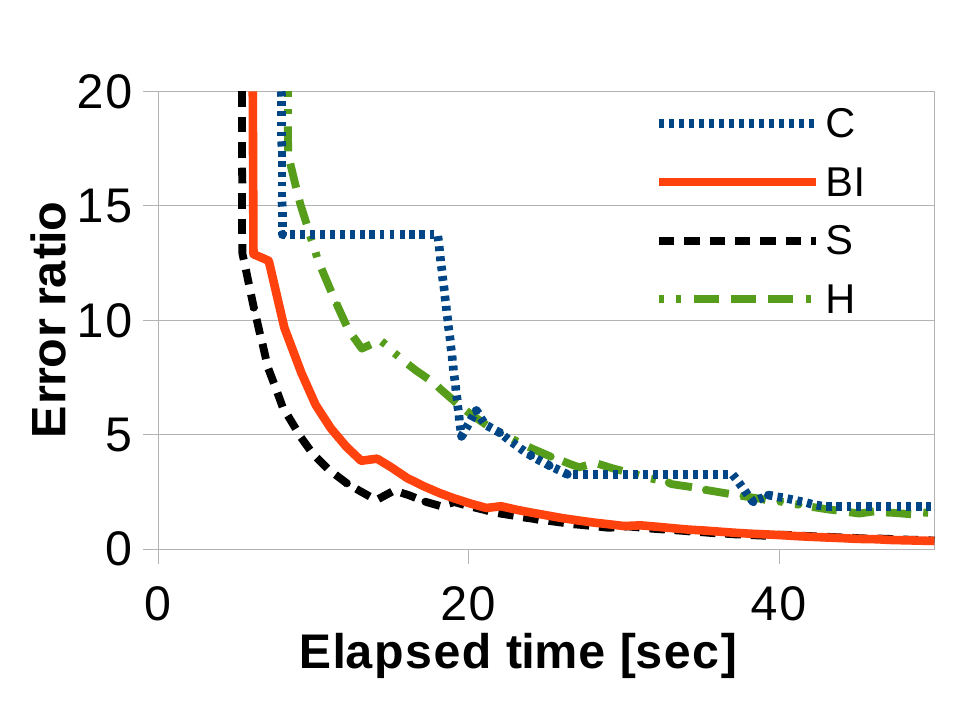}\label{fig:comparison-t4}}
		\subfloat[]{\includegraphics[width=2.07in,height=1.75in]{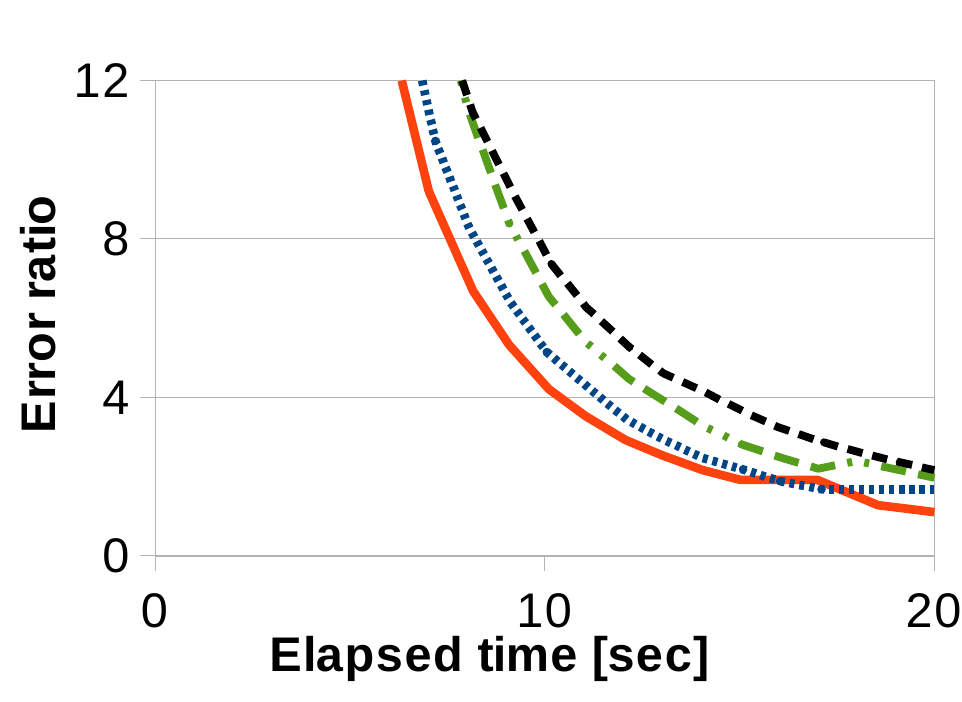}\label{fig:comparison-t16}}
		\caption{Sampling comparison as a function of the number of threads: 1 thread (a), 4 threads (b), and 16 threads (c).}
		\label{fig:comparison}
	\end{center}
\end{figure*}
%%%%%%%%%%%%%%%%%%%%%%%%%%%%%%%%%%%%%

%%%%%%%%%%%%%%%%%%%%%%%%%%%%%%%%%%%%
\subsubsection{Sample Synopsis}\label{sssec:experiments:results:synopsis}

The effect of the sample synopsis on the execution of a query sequence is depicted in Figure~\ref{fig:sequence-increase} for the \texttt{synthetic} dataset. We use the resource-aware bi-level sampling \texttt{BI}. We execute 10 instances of the same query at 5 increasing accuracy requirements---each query is executed twice with the same accuracy. The execution time to achieve the required accuracy, the ratio of sampled chunks, and the ratio of sampled tuples are computed for two synopsis sizes---16 MB and 48 MB. These store samples of $0.08\%$ and $0.24\%$ ratios from the entire dataset. The execution time for the first query is considerably larger due to accessing almost all the chunks. Since the subsequent queries can be answered based on the synopsis, their execution is much faster. In the best case, the query can be answered exclusively from the synopsis and no access to the disk is necessary. This happens for the large synopsis. The small synopsis has to expel tuples from sampled chunks in order to make space for new chunks. As a result, as the accuracy requirement increases, the same chunk has to be accessed multiple times from disk. This results in a significant increase in the execution time.

%%%%%%%%%%%%%%%%%%%%%%%%%%%%%%%%%%%%
\begin{figure*}[htbp]
	\begin{center}
		\includegraphics[width=0.99\textwidth]{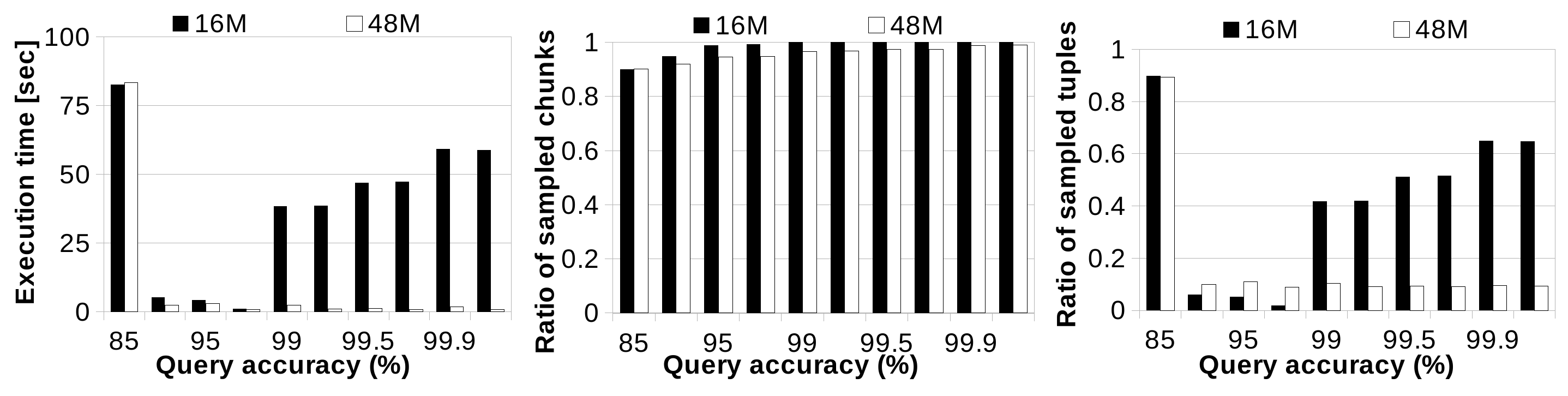}
		\caption{Sample synopsis effect on a sequence of queries with increasing accuracy. For each accuracy level, the query is executed two times.}
		\label{fig:sequence-increase}
	\end{center}
\end{figure*}
%%%%%%%%%%%%%%%%%%%%%%%%%%%%%%%%%%%%%

Figure~\ref{fig:sequence-decrease} depicts the impact of the sample synopsis on a sequence of queries with decreasing accuracy requirement. The setup is exactly the same as in Section~\ref{sssec:experiments:results:synopsis} with the difference that the accuracy is in decreasing order. In this case, the large synopsis is built entirely during the first query and all the subsequent queries are answered exclusively from the synopsis. This is the best situation for using the synopsis. The small synopsis does not contain sufficient tuples to support estimation for queries with accuracy larger than $95\%$. As a result, access to the raw data is necessary to achieve the required accuracy. This is reflected in the larger execution time and the larger number of tuples that are inspected. Even in this case, though, the synopsis provides a significant reduction in the execution time.

%%%%%%%%%%%%%%%%%%%%%%%%%%%%%%%%%%%%
\begin{figure*}[htbp]
	\begin{center}
		\includegraphics[width=.95\textwidth,height=1.82in]{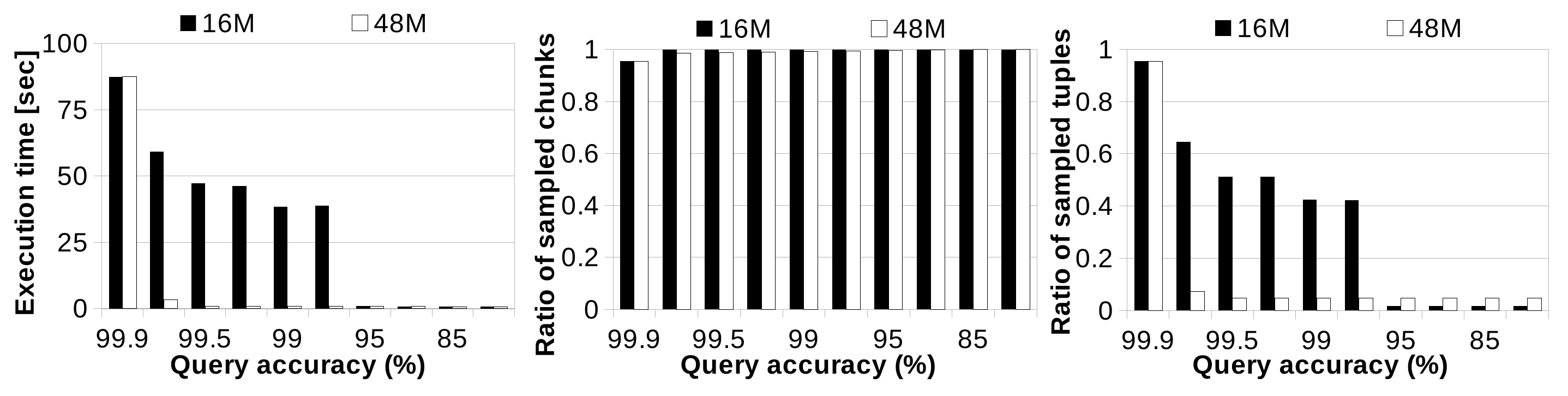}
		\caption{Sample synopsis effect on a sequence of queries with decreasing accuracy. For each accuracy, the query is executed two times.}
		\label{fig:sequence-decrease}
	\end{center}
\end{figure*}
%%%%%%%%%%%%%%%%%%%%%%%%%%%%%%%%%%%%%

%%%%%%%%%%%%%%%%%%%%%%%%%%%%%%%%%%%%
\subsubsection{Resource Utilization}\label{sssec:experiments:results:resource}

Figure~\ref{fig:resource-utilization} depicts the CPU and I/O utilization for resource-aware bi-level and chunk-level sampling, respectively, in a CPU-bound scenario. In this case, the earlier a chunk is finished, the better. \texttt{BI} achieves this by extracting only the necessary number of tuples. This results in a much better I/O utilization than \texttt{C}. Since \texttt{C} has to process all the tuples in the chunk, it blocks chunk reading completely until CPUs become available. Essentially, the I/O throughput transitions periodically between two extremes based on how chunks finish processing. The processing time for a chunk is more random in the case of \texttt{BI} because it depends on the accuracy rather than the number of tuples.

%%%%%%%%%%%%%%%%%%%%%%%%%%%%%%%%%%%%
\begin{figure*}[htbp]
	\begin{center}
		\subfloat[]{\includegraphics[width=0.48\textwidth]{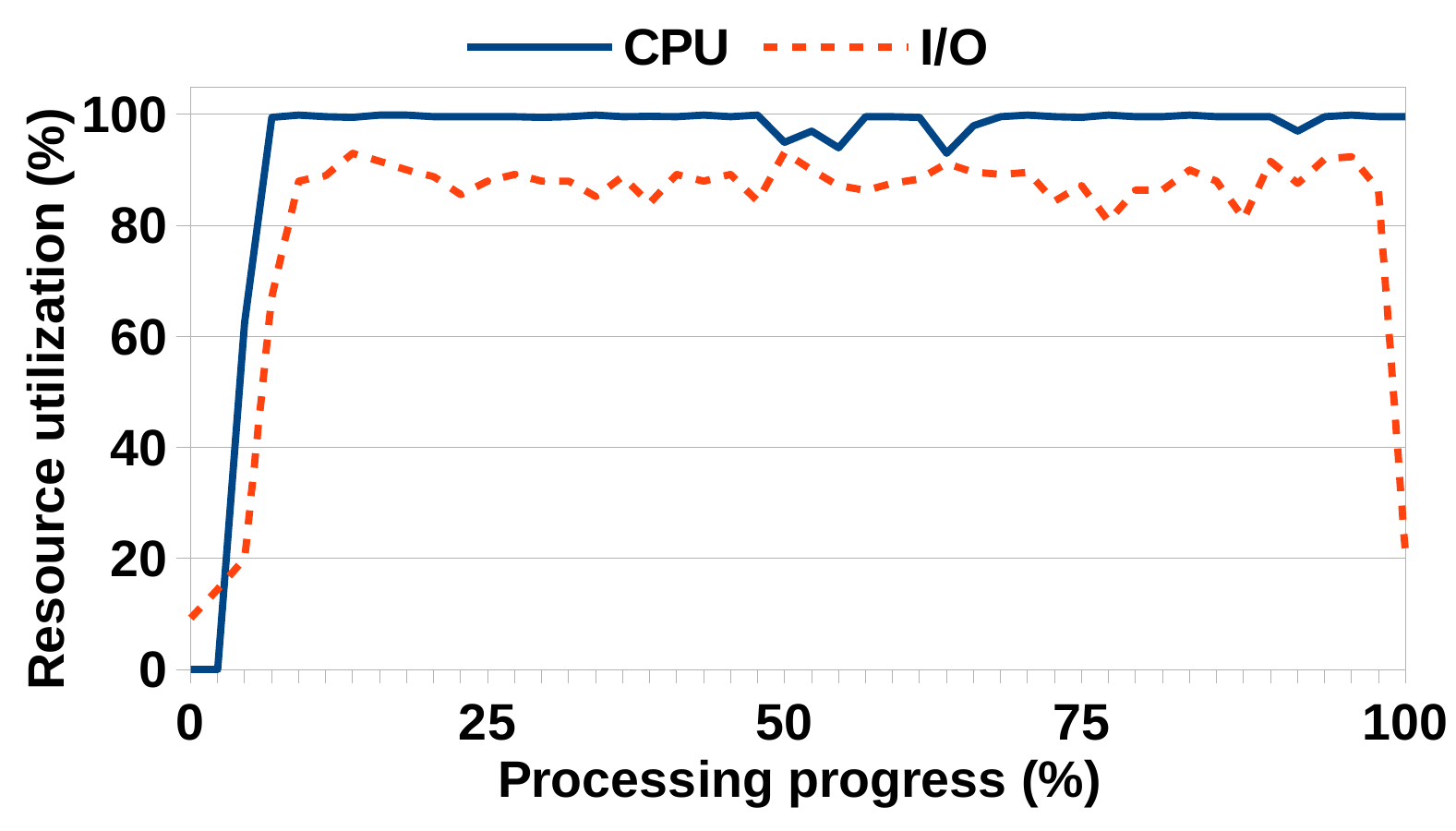}\label{fig:res-olaraw}}
		\subfloat[]{\includegraphics[width=0.48\textwidth]{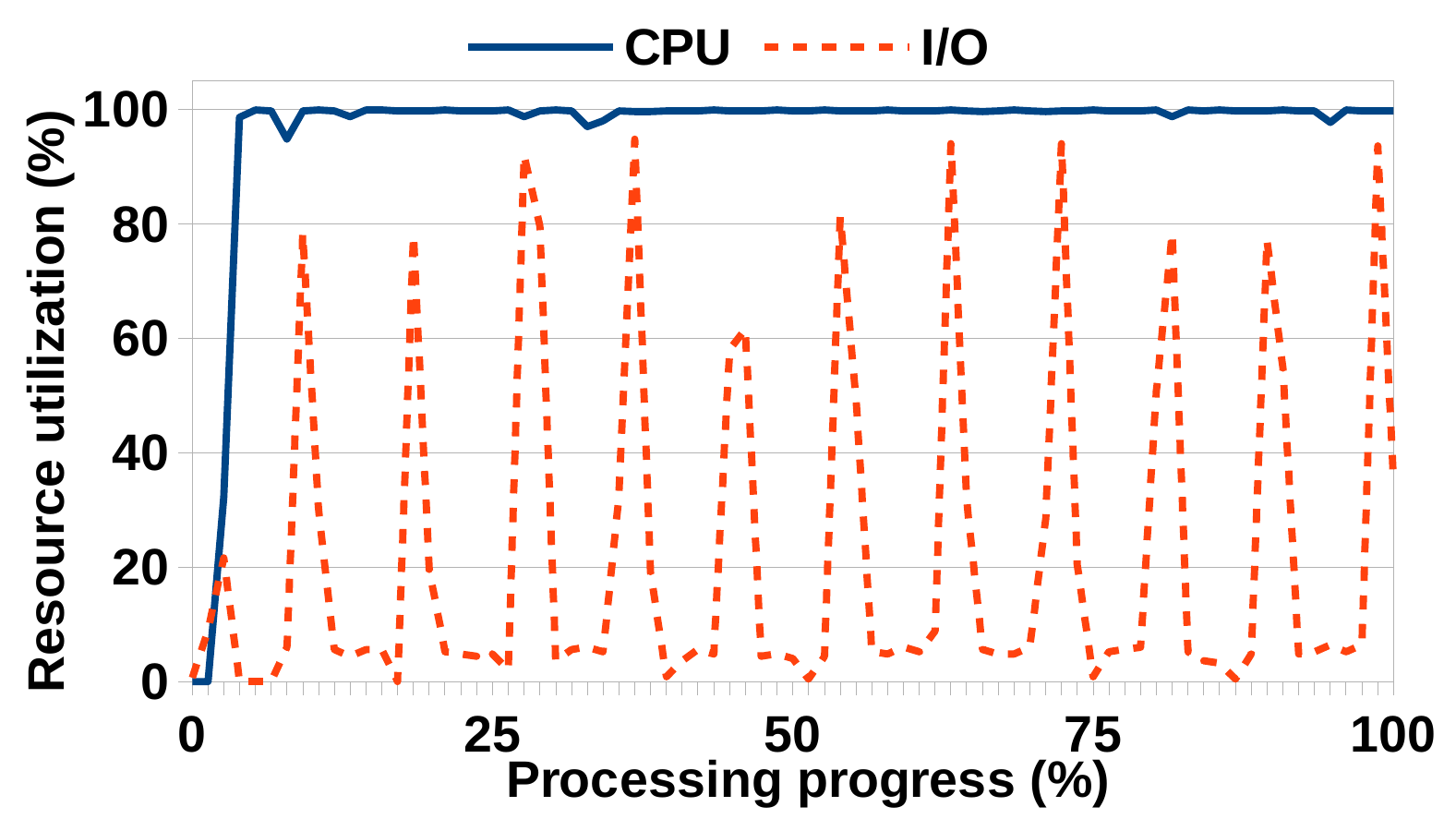}\label{fig:res-cluster}}
		\caption{CPU and I/O utilization for resource-aware bi-level (a) and chunk-level (b) sampling.}
		\label{fig:resource-utilization}
	\end{center}
\end{figure*}
%%%%%%%%%%%%%%%%%%%%%%%%%%%%%%%%%%%%%

%%%%%%%%%%%%%%%%%%%%%%%
\subsubsection{Confidence Bounds Correctness}\label{conf-bounds}

In order to verify that the confidence bounds produced for bi-level sampling are correct, we execute a series of Monte Carlo simulations over the \texttt{synthetic} dataset at $50\%$ selectivity. We include chunk-level sampling in the simulations as well, however, we do not enforce chunk reordering for estimation. This lets chunk-level sampling vulnerable to the inspection paradox. We facilitate the conditions for this to happen by setting the number of threads to $4$ because in this regimen processing transitions between CPU-bound and I/O-bound.

Table~\ref{tbl:exp:montecarlo} contains the results of $100$ simulations. We measure the number of times the correct result is within the computed bounds at $95\%$ accuracy after a certain fraction of chunks are processed. The results show that the bounds for bi-level sampling contain the true query answer in $95$ or more of the runs. The only exception is for a very small number of chunks when bound derivation is difficult because the heterogeneity between chunks cannot be accurately assess. In the case of chunk-level sampling without reordering, the effect of the inspection paradox is immediately clear. The more chunks are processed, more chances for inversions to take place. This results in inconsistent bounds. Nonetheless, once a sufficiently large number of chunks are processed, the bounds stabilize---not near the expected value, though. 

%%%%%%%%%%%%%%%%%%%%%%%%%%%%%%%%%%
\begin{table}[htbp]
  \begin{center}
    \begin{tabular}{l|rrrrrrr}

     & \multicolumn{7}{c}{\textbf{Fraction of processed chunks}} \\

	 & .02 & .03 & .04 & .05 & .10 & .20 & .30 \\
	\hline
	
	Chunk-level & .92 & .89 & .82 & .78 & .75 & .78 & .71\\
	Bi-level & .94 & .96 & .98 & .95 & .97 & .98 & .99\\

    \end{tabular}
  \end{center}

\caption{Fraction of $95\%$ accuracy confidence bounds that contain the correct query result.}\label{tbl:exp:montecarlo}
\end{table}
%%%%%%%%%%%%%%%%%%%%%%%%%%%%%%%%%%

%%%%%%%%%%%%%%%%%%%%%%%%%%%%%%%%%%%%
\subsection{Discussion}\label{ssec:experiments:discussion}

The experimental results confirm that OLA-RAW bi-level sampling outperforms external tables and chunk-level sampling. In the best case, OLA-RAW achieves the required accuracy in as little as $10\%$ of the time required by external tables to answer the query exactly and chunk-level sampling to achieve the same accuracy. This happens in a heavily CPU-bound setting where the bottleneck is the number of processed tuples. The worst scenario for OLA-RAW is a low selectivity query in an I/O-bound context with a sufficiently large number of threads. Due to the sampling overhead, both bi-level and chunk-level sampling are slower than external tables. The results also prove that resource-aware bi-level sampling is the most adaptable scheme to the processing environment while a sample synopsis of less than $1\%$ in size can provide a reduction of more than $10X$ in execution time across a sequence of correlated queries.

\section{RELATED WORK}\label{sec:rel-work}

Our main contribution is to integrate online aggregation in raw data processing. Thus, these two lines of research -- raw data processing and online aggregation -- are the most relevant to our work. We discuss the relationship between OLA-RAW bi-level sampling and these two research areas in the following.

%%%%%%%%%%%%%%%%%%%%%%%%%%%%%%%%%%%%%%%%%%
\textbf{Raw data processing.}
At a high level, we can group raw data processing into two categories. In the first category, we have extensions to traditional database systems that allow raw file processing inside the execution engine. Examples include external tables~\cite{oracle:external-tables,datastage,impala} and various optimizations that eliminate the requirement for scanning the entire file to answer the query~\cite{files-queries-results,nodb,data-vaults}. Modern database engines -- e.g., Oracle, MySQL, Impala -- provide external tables as a feature to directly query flat files using SQL without paying the upfront cost of loading the data into the system. NoDB~\cite{nodb} enhances external tables by extracting only the attributes required in the query and caching them in memory for use in subsequent queries. Data vaults~\cite{data-vaults} and SDS/Q~\cite{blanas:sigmod-2014} apply the same idea of query-driven just-in-time caching to scientific repositories. While the proposed sample synopsis inherits in-memory caching, it caches samples rather than full columns. Adaptive partial loading~\cite{files-queries-results} materializes the cached data in NoDB to secondary storage before query execution starts---loading is query-driven. SCANRAW~\cite{scanraw,scanraw:tods} is a super-scalar adaptive external tables implementation that materializes data only when I/O resources are available. Instant loading~\cite{instant-loading} introduces vectorized SIMD implementations for tokenizing. RAW~\cite{raw} and its extensions VIDa~\cite{vida,vida-llvm} generate \texttt{EXTRACT} operators just-in-time for the underlying data file and the incoming query. The second category is organized around Hadoop MapReduce which processes natively raw data by including the \texttt{EXTRACT} code in the Map and Reduce functions. Invisible loading~\cite{invisible-loading} focuses on eliminating the \texttt{EXTRACT} code by loading the already converted data inside a database. While similar to adaptive partial loading, instead of saving all the tuples into the database, only a fraction of tuples are stored for every query. None of these solutions supports sampling over raw data and estimation---the central contribution of this work.

\textbf{Online aggregation.}
The database online aggregation literature has its origins in the seminal paper by Hellerstein et al.~\cite{ola}. We can broadly categorize this body of work into system design~\cite{control,dbo,demo:dbo,turbo:dbo}, online join algorithms~\cite{ripple-join,sms-join,pr-join,wander-join}, online algorithms for estimations other than join~\cite{ola-set,ola-extreme,feifei:spatial-ola}, and methods to derive confidence bounds~\cite{haas-bounds}. The distributed online aggregation literature is not as rich. Luo et al.~\cite{par-hash-ripple} extend the ripple join algorithm~\cite{ripple-join} to a distributed setting. Wu et al.~\cite{distributed-ola} extend online aggregation to distributed P2P networks. They introduce a synchronized sampling estimator over partitioned data that requires data movement from storage nodes to processing nodes. In subsequent work, Wu et al.~\cite{continuous-sampling} tackle online aggregation over multiple queries. The third piece of relevant work is online aggregation in MapReduce (Hadoop or Spark). BlinkDB~\cite{blink,agarwal:bootstrap} implements a multi-stage approximation mechanism based on pre-computed sampling synopses of multiple sizes, while EARL~\cite{earl} and ABS~\cite{zeng:bootstrap} use bootstrapping to produce multiple estimators from the same sample. iOLAP~\cite{iolap} models online aggregation as incremental view maintenance with uncertainty propagation. Sample+Seek~\cite{sample-seek} introduces measure-biased sampling to provide error guarantees for \texttt{GROUP BY} queries with many groups. Quickr~\cite{quickr} injects single-pass samplers in query execution plans to generate one-time approximate results. With almost no exceptions, all of this work is based on tuple-based sampling. The inadequacy of this type of sampling for database processing has been recognized in~\cite{chaudhuri:block-sampling} where cardinality estimators based on chunk-level sampling are introduced. This type of sampling is later applied to parallel online aggregation in Hadoop~\cite{hadoop-online,online-mapreduce}. The only application of bi-level Bernoulli sampling to database processing is given in~\cite{bi-level-bernoulli}. The last two solutions are the closest to OLA-RAW and they are discussed in detail throughout the paper. Our main contribution to online aggregation is the design of parallel bi-level sampling estimators that avoid the inspection paradox.

\section{CONCLUSIONS AND FUTURE WORK}\label{sec:conclusions}

In this paper, we present OLA-RAW, a bi-level sampling scheme for parallel online aggregation over raw data. OLA-RAW sampling is query-driven and performed exclusively in-situ during query execution, without data reorganization. In order to avoid the expensive conversion cost, OLA-RAW builds and maintains incrementally a memory-resident bi-level sample synopsis. We implement OLA-RAW inside a modern in-situ data processing system and evaluate its performance across several real and synthetic datasets and file formats. Our results confirm that OLA-RAW bi-level sampling outperforms external tables and chunk-level sampling -- by as much as $10X$ -- and leads to a focused data exploration process that avoids unnecessary work and discards uninteresting data. In future work, we plan to perform a thorough investigation of the estimator sensitivity to the chunk size. This is a well-known problem for bi-level sampling. Adaptive solutions that change the chunk size dynamically at runtime are an interesting direction to pursue.

\begin{comment}
In this paper, we propose \textit{OLA-RAW} -- a novel system for in-situ processing over raw files that integrates online aggregation and raw file processing seamlessly while preserving their advantages. \textit{OLA-RAW} supports single-query optimal execution with a \textit{parallel super-scalar pipeline implementation} that overlaps data reading, conversion into the database representation, and query processing. \textit{OLA-RAW} implements \textit{sample maintenances} as a gradual loading mechanism to store converted data inside the database. We implement \textit{OLA-RAW} in a state-of-the-art in-situ data processing system and evaluate its performance across a variety of different datasets and file formats. Our results show that \textit{OLA-RAW} achieves optimal performance for a query sequence at any point in the processing. 

In the future works, we plan to focus on extending \textit{OLA-RAW} with support for multi-query processing over raw files.
\end{comment}

{\small
\bibliographystyle{abbrv}
%\bibliography{biblio}  % vldb_sample.bib is the name of the 

}

%\appendix

%\input{appendix-proofs}

%\input{appendix-experiments}

\end{document}